\documentclass{article}[11pt]
\usepackage[utf8]{inputenc}
\usepackage{style}

\title{Near-Optimal UGC-hardness of Approximating {\sc Max $k$-CSP$_R$}}
\author{Pasin Manurangsi}
\author{Preetum Nakkiran}
\author{Luca Trevisan}
\affil{University of California, Berkeley \authorcr
		  \texttt{\small \{pasin,preetum,luca\}@berkeley.edu}}

\date{}

\begin{document}
\maketitle

\newcommand{\fk}{f^{\leq d}}
\newcommand{\Fk}{F^{\leq d}}
\newcommand{\Gk}{G^{\leq d}}
\newcommand{\gk}{g^{\leq d}}
\renewcommand{\t}{\widetilde}
\newcommand{\pcorr}{\overset{\rho}{\gets}}
\newcommand{\analog}[1]{\{#1\}}

\renewcommand{\i}[2]{{#1}^{(#2)}}

\abstract{In this paper, we prove an almost-optimal hardness for {\sc Max
    $k$-CSP$_R$} based on Khot's Unique Games Conjecture (UGC). In Max
    $k$-CSP$_R$, we are given a set of predicates each of which depends on
    exactly $k$ variables.
    Each variable can take any value from $1, 2, \dots, R$.
    The goal is to find an assignment to variables that maximizes the number of satisfied predicates.

Assuming the Unique Games Conjecture, we show that it is NP-hard to approximate {\sc Max $k$-CSP$_R$} to within factor $2^{O(k \log k)}(\log R)^{k/2}/R^{k - 1}$ for any $k, R$.
To the best of our knowledge, this result improves on all the known hardness of approximation results when
$3 \leq k = o(\log R/\log \log R)$.
In this case, the previous best hardness result was
NP-hardness of approximating within a factor $O(k/R^{k-2})$ by Chan.
When $k = 2$, our result matches the best known UGC-hardness result of Khot, Kindler, Mossel and O'Donnell.

In addition, by extending an algorithm for {\sc Max 2-CSP$_R$} by Kindler, Kolla and Trevisan, we provide an $\Omega(\log R/R^{k - 1})$-approximation algorithm for {\sc Max $k$-CSP$_R$}. This algorithm implies that our inapproximability result is tight up to a factor of $2^{O(k \log k)}(\log R)^{k/2 - 1}$. In comparison, when $3 \leq k$ is a constant, the previously known gap was $O(R)$, which is significantly larger than our gap of $O(\polylog R)$.

Finally, we show that we can replace the Unique Games Conjecture assumption with Khot's $d$-to-1 Conjecture and still get asymptotically the same hardness of approximation.
}

\newpage


\newpage

\setlength{\parindent}{0pt}
\setlength{\parskip}{6pt plus 2pt minus 1pt}

\section{Introduction}

Maximum Constraint Satisfaction Problem ({\sc Max CSP}) is an optimization
problem where the inputs are a set of variables, an alphabet set, and a
set of predicates. Each variable can be assigned any alphabet from the
alphabet set and each predicate depends only on the assignment to a subset of variables. The goal is to find an assignment to the variables that maximizes the number of satisfied predicates.

Many natural optimization problems, such as {\sc Max Cut}, {\sc Max $k$-CUT} and {\sc Max $k$-SAT}, can be formulated as {\sc Max CSP}. In addition, {\sc Max CSP} has been shown to help approximate other seemingly-unrelated problems such as {\sc Densest $k$-Subgraph}~\cite{CHK11}. Due to this, {\sc Max CSP} has long been researched by the approximation algorithm community~\cite{Tre98,Hast05,CMM09,MM14,KKT15,GM15}. Furthermore, its relation to PCPs ensures that {\sc Max CSP} is also well-studied on the hardness of approximation side~\cite{ST00,Eng05,ST06,KKMO07,AM08,GR08,EH08,Chan13}.

The main focus of this paper is on {\sc {Max $k$-CSP}$_R$}, a family of {\sc Max
CSP} where the alphabet set is of size $R$ and each predicate depends on only
$k$ variables. On the hardness of approximation side, most early works focused
on boolean {\sc {Max $k$-CSP}}. Samorodnitsky and Trevisan first showed that
{\sc {Max $k$-CSP}$_2$} is NP-hard to approximate to within factor
$2^{O(\sqrt{k})}/2^{k}$~\cite{ST00}. Engebretsen and Holmerin later improved
constant factors in the exponent $O(\sqrt{k})$ but still yielded hardness of a
factor $2^{O(\sqrt{k})}/2^{k}$~\cite{EH08}. To break this barrier, Samorodnitsky
and Trevisan proved a hardness of approximation conditioned on Khot's Unique
Games Conjecture (UGC), which will be discussed in more detail later; they
achieved a ratio of $O(k/2^k)$ hardness, which is tight up to a constant for the boolean case~\cite{ST06}. Chan later showed that NP-hardness of approximation at this ratio can be achieved unconditionally and, thus, settled down the
approximability of {\sc {Max $k$-CSP}$_2$}~\cite{Chan13}.

Unlike the boolean case, the approximability of {\sc {Max $k$-CSP}$_R$} when $R
> 2$ is still not resolved. In this case, Engebretsen showed
$R^{O(\sqrt{k})}/R^k$ NP-hardness of approximation~\cite{Eng05}. Under the
Unique Games Conjecture, a hardness of approximation of $O(kR/R^{k-1})$ factor
was proven by Austrin and Mossel~\cite{AM08} and, independently, by Guruswami
and Raghavendra~\cite{GR08}.
For the case $k=2$, results by Khot et al. \cite{KKMO07} implicitly
demonstrate UGC-hardness of approximation within $O(\log R / R)$,
made explicit in~\cite{KKT15}.
Moreover, Austrin and Mossel proved UGC-hardness of
approximation of $O(k/R^{k-1})$ for infinitely many $k$s~\cite{AM08},
but in the regime $k \geq R$.
Recently, Chan was able to remove the Unique Game Conjecture assumption from these
results~\cite{Chan13}. More specifically, Chan showed NP-hardness of
approximation of factor $O(kR/R^{k-1})$ for every $k, R$ and that of
factor $O(k/R^{k-1})$ for every $k \geq R$. Due to an approximation algorithm with matching approximation ratio by Makarychev and Makarychev~\cite{MM14}, Chan's result established tight hardness of approximation for $k \geq R$. On the other hand, when $k < R$, Chan's result gives $O(kR/R^{k-1})$ hardness of approximation whereas the best known approximation algorithm achieves only $\Omega(k/R^{k - 1})$ approximation ratio~\cite{MM14,GM15}.
In an attempt to bridge this gap, we prove the following theorem.

\begin{theorem}[Main Theorem] \label{thm:main}
Assuming the Unique Games Conjecture,
it is NP-hard to approximate {\sc {Max $k$-CSP}$_R$} to within $2^{O(k \log
k)}(\log R)^{k/2}/R^{k - 1}$ factor,
for any $k \geq 2$ and any sufficiently large $R$.
\end{theorem}

\begin{figure}[h]
\centering
\begin{tabular}{c | c | c | c | c}
Range of $k, R$ & NP-Hardness & UGC-Hardness & Approximation & References \\
\hline
$k = 2$ & $O\left(\frac{\log R}{\sqrt{R}}\right)$ & $O\left(\frac{\log R}{R}\right)$ & $\Omega\left(\frac{\log R}{R}\right)$ & \cite{Chan13,KKMO07,KKT15} \\
$3 \leq k < R$ & $O\left(\frac{k}{R^{k-2}}\right)$ & -- & $\Omega\left(\frac{k}{R^{k-1}}\right)$ & \cite{Chan13,MM14,GM15} \\
$R \leq k$ & $O\left(\frac{k}{R^{k-1}}\right)$ & -- & $\Omega\left(\frac{k}{R^{k-1}}\right)$ & \cite{Chan13,MM14} \\
\hline
Any $k, R$ & -- & $\frac{2^{O(k \log k)}(\log R)^{k/2}}{R^{k - 1}}$  & $\Omega\left(\frac{\log R}{R^{k - 1}}\right)$ & this work
\end{tabular}
\caption{Comparison between our work and previous works. We list the previous
best known results alongside our results. From  previous works, there is either
an NP-hardness or a UGC-hardness result matching the best known approximation algorithm in every case except when $3 \leq k < R$.
Our hardness result improves on the best known hardness result when $k = o(\log
R/\log \log R)$,
and our approximation algorithm improves on the previously known algorithm when $k = o(\log R)$.}
\end{figure}

When $k = o(\log R/ \log \log R)$, our result
improves upon the previous best known hardness of approximation result in this regime,
due to Chan.
In particular, when $k$ is constant, our results are tight up to a factor of
$O(\text{polylog } R)$.  While Chan's results hold unconditionally, our result, similar
to many of the aforementioned results (e.g. \cite{ST06,AM08,GR08}), rely on the
Unique Games Conjecture.

A {\em unique game} is a {\sc Max 2-CSP} instance where each constraint is a
permutation. The {\em Unique Games Conjecture (UGC)}, first proposed by Khot in
his seminal paper in 2002~\cite{Khot02}, states that, for any sufficiently small
$\eta, \gamma > 0$, it is NP-hard to distinguish a unique game where at least $1
- \eta$ fraction of constraints can be satisfied from a unique game where at
most $\gamma$ fraction of constraints can be satisfied. The UGC has since made a
huge impact in hardness of approximation; numerous hardness of approximation
results not known unconditionally can be derived assuming the UGC. More
surprisingly, UGC-hardness of approximation for various problems, such as {\sc
Max Cut}~\cite{KKMO07}, {\sc Vertex Cover}~\cite{KR08} and {\em any} {\sc Max
CSP}~\cite{Rag08}\footnote{Raghavendra showed in~\cite{Rag08} that it is hard to
approximate any {\sc Max CSP} beyond what a certain type of semidefinite program
can achieve. However, determining the approximation ratio of a semidefinite
program is still not an easy task. Typically, one still needs to find an
integrality gap for such a program in order to establish the approximation ratio.}, are known to be tight. For more details on UGC and its implications, we refer interested readers to Khot's survey~\cite{Khot10} on the topic.

Another related conjecture from~\cite{Khot02} is the {\em $d$-to-1 Conjecture}.
In the $d$-to-1 Conjecture, we consider {\em $d$-to-1 games} instead of unique
games. A $d$-to-1 game is an instance of {\sc Max 2-CSP} where the constraint
graph is bipartite.
Moreover, each constraint must be a $d$-to-1 function, i.e., for each assignment
to a variable on the right, there exists $d$ assignments to the corresponding variable on the
left that satisfy the constraint. The $d$-to-1 Conjecture states that, for any
sufficiently small $\gamma > 0$, it is NP-hard to distinguish between a fully
satisfiable $d$-to-1 game and a $d$-to-1 game where at most $\gamma$ fraction of constraints can be satisfied. Currently, it is unknown whether the $d$-to-1 Conjecture implies the Unique Games Conjecture and vice versa.

While the $d$-to-1 Conjecture has yet to enjoy the same amount of influence as
the UGC, it has been proven successful in providing a basis for hardness of graph coloring problems ~\cite{DMR09,DS10,GS11} and for {\sc Max $3$-CSP} with perfect completeness~\cite{OW09,Tan09}. Here we show that, by assuming the $d$-to-1 Conjecture instead of UGC, we can get a similar hardness of approximation result for {\sc Max $k$-CSP$_R$} as stated below.

\begin{theorem} \label{thm:d-to-1}
Assuming the $d$-to-1 Games Conjecture holds for some $d$,
it is NP-hard to approximate {\sc {Max $k$-CSP}$_R$} to within $2^{O(k \log
k)}(\log R)^{k/2}/R^{k - 1}$ factor,
for any $k \geq 2$ and any sufficiently large $R$.
\end{theorem}

As mentioned earlier, there has also been a long line of works in approximation
algorithms for {\sc Max $k$-CSP$_R$}. In the boolean case, Trevisan first showed
a polynomial-time approximation algorithm with approximation ratio
$2/2^k$~\cite{Tre98}. Hast later improved the ratio to $\Omega(k/(2^k \log
k))$~\cite{Hast05}. Charikar, Makarychev and Makarychev then came up with an
$\Omega(k/2^k)$-approximation algorithm~\cite{CMM09}. As stated when discussing hardness of approximation of {\sc Max $k$-CSP$_2$}, this approximation ratio is tight up to a constant factor.

For the non-boolean case, Charikar, Makarychev, and Makarychev's algorithm
achieve $\Omega(k\log R/R^k)$ ratio for {\sc Max $k$-CSP$_R$}. Makarychev, and
Makarychev later improved the approximation ratio to $\Omega(k/R^{k - 1})$ when
$k = \Omega(\log R)$~\cite{MM14}. Additionally, the algorithm was extended by
Goldshlager and Moshkovitz to achieve the same approximation ratio for any $k, R$~\cite{GM15}. On this front, we show the following result.

\begin{theorem} \label{thm:approx}
There exists a polynomial-time $\Omega(\log R/R^{k - 1})$-approximation algorithm for {\sc Max $k$-CSP$_R$}.
\end{theorem}

In comparison to the previous algorithms, our algorithm gives better approximation ratio than all the known algorithms when $k = o(\log R)$. We remark that our algorithm is just a simple extension of Kindler, Kolla and Trevisan's polynomial-time $\Omega(\log R/R)$-approximation algorithm for Max $2$-CSP$_R$~\cite{KKT15}.

\subsection{Summary of Techniques}
Our reduction from Unique Games to {\sc Max $k$-CSP$_R$} follows the reduction of
\cite{KKMO07} for {\sc Max $2$-CSP}s. We construct a $k$-query PCP using
a Unique-Label-Cover ``outer verifier'',
and then design a $k$-query Long Code test as an ``inner verifier''.
For simplicity, let us think of $k$ as a constant. We essentially construct a $k$-query
\emph{Dictator-vs.-Quasirandom} test
for functions $f: [R]^n \to [R]$,
with completeness $\Omega(1/(\log{R})^{k/2})$
and soundness $O(1/R^{k-1})$.
Our test is structurally similar to the $2$-query ``noise stability'' tests of
\cite{KKMO07}:
first we pick a random $z \in [R]^n$, then
we pick $k$ weakly-correlated queries $\i{x}{1}, \dots, \i{x}{k}$
by choosing each $\i{x}{i} \in [R]^n$ as a noisy copy of $z$, i.e.,
each coordinate $(\i{x}{i})_j$ is chosen as $z_j$
with some probability $\rho$ or uniformly at random otherwise.
We accept iff $f(\i{x}{1}) = f(\i{x}{2}) = \dots = f(\i{x}{k})$.
The key technical step is our analysis of the soundness of this test.
We need to show that if $f$ is a balanced function with
small low-degree influences,
then the test passes with probability $O(1/R^{k-1})$.
Intuitively, we would like to say that for high enough noise,
the values $f(\i{x}{i})$ are roughly independent and uniform, so the test
passes with probability around $1/R^{k-1}$.
To formalize this intuition, we utilize the \emph{Invariance Principle} and
\emph{Hypercontractivity}.

More precisely, if we let $f^i(x): [R]^n \to \{0, 1\}$ be the indicator function for $f(x) = i$,
then the test accepts iff
$f^i(\i{x}{1}) = \dots = f^i(\i{x}{k}) = 1$ for some $i \in [R]$.
For each $i \in [R]$, this probability can be written as the expectation of the product:
$\E[f^i(\i{x}{1})f^i(\i{x}{2})\dots f^i(\i{x}{k})]$.
Since $\i{x}{i}$'s are chosen as noisy copies of $z$, this expression is related to the
$k$-th norm of a noisy version of $f^i$.
Thus, our problem is reduced to bounding the $k$-norm of a noisy function
$\t f^i: [R]^n \to [0, 1]$, which has bounded one-norm $\E[\t f^i] = 1/R$ since
$f$ is balanced. To arrive at this bound, we first apply the Invariance Principle,
which essentially states that a low-degree low-influence function on $[R]^n$ behaves
on random input similarly to its ``boolean analog'' over domain
$\{\pm 1\}^{nR}$. Here ``boolean analog'' refers to the function over
$\{\pm 1\}^{nR}$ with matching Fourier coefficients.

Roughly speaking, now that we have transfered to the boolean domain, we are
left to bound the $k$-norm of a noisy function on $\{\pm 1\}^{nR}$ based on its
one-norm. This follows from Hypercontractivity in the boolean setting,
which bounds a higher norm of any noisy function on boolean domain in terms
of a lower norm.

The description above hides several technical complications.
For example, when we pass from a function $[R]^n \to [0, 1]$ to its
``boolean analog'' $\{\pm 1\}^{nR} \to \R$,
the range of the resulting function is no longer bounded to $[0, 1]$.
This, along with the necessary degree truncation,
means we must be careful when bounding norms.
Further, Hypercontractivity only allows us to pass from $k$-norms to
$(1+\epsilon)$-norms for small $\epsilon$, so we cannot use
the known $1$-norm directly.
For details on how we handle these issues, see Section~\ref{sec:inapprox}.
This allows us to prove soundness of our dictator test without
passing through results on Gaussian space,
as was done to prove the ``Majority is Stablest'' conjecture \cite{MOO10}
at the core of the \cite{KKMO07} 2-query dictator test.

To extend our result to work with $d$-to-1 Games Conjecture in place of UGC, we observe that our proof goes through even when we assume a conjecture weaker than the UGC, which we name the {\em One-Sided Unique Games Conjecture}. The only difference between the original UGC and the One-Sided UGC is that the completeness in UGC is allowed to be any constant smaller than one but the completeness is a fixed constant for the One-Sided UGC. The conjecture is formalized as Conjecture~\ref{conj:OSUGC}. We show that the $d$-to-1 Games Conjecture also implies the One-Sided UGC, which means that our inapproximability result for {\sc Max $k$-CSP$_R$} also holds when we change our assumption to $d$-to-1 Games.

Lastly, for our approximation algorithm, we simply extend the Kindler, Kolla and Trevisan's algorithm by first creating an instance of {\sc Max 2-CSP$_R$} from {\sc Max $k$-CSP$_R$} by projecting each constraint to all possible subsets of two variables. We then use their algorithm to approximate the instance. Finally, we set our assignment to be the same as that from KKT algorithm with some probability. Otherwise, we pick the assignment uniformly at random from $[R]$. As we shall show in Section~\ref{sec:approx}, with the right probability, this technique can extend not only the KKT algorithm but any algorithm for {\sc Max $k'$-CSP$_R$} to an algorithm for {\sc Max $k$-CSP$_R$} where $k > k'$ with some loss in the advantage over the naive randomized algorithm.

\subsection{Organization of the Paper}

In Section~\ref{sec:prelim}, we define notations and list background knowledge
that will be used throughout the paper. Next, we prove our hardness of
approximation results, i.e., Theorem~\ref{thm:main} and
Theorem~\ref{thm:d-to-1}, in Section~\ref{sec:inapprox}.
In Section~\ref{sec:approx}, we show how to extend Kindler et al.'s algorithm to
{\sc Max $k$-CSP$_R$} and prove Theorem~\ref{thm:approx}.
We also explicitly present the dictator test that is implicit in our hardness
proof, in Section~\ref{sec:dtest}.
Finally, in Section~\ref{sec:lim}, we discuss interesting open questions and directions for future works.

\section{Preliminaries} \label{sec:prelim}

In this section, we list notations and previous results that will be used to prove our results.

\subsection{{\sc Max $k$-CSP$_R$}}

We start by giving a formal definition of {\sc Max $k$-CSP$_R$}, which is the main focus of our paper.

\begin{definition}[{\sc Max $k$-CSP$_R$}]
An instance $(\mathcal{X}, \mathcal{C})$ of (weighted) {\sc Max $k$-CSP$_R$} consists of
\begin{itemize}
\item A set $\mathcal{X} = \{x_1, \dots, x_n\}$ of variables.
\item A set $\mathcal{C} = \{C_1, \dots, C_m\}$ of constraints. Each constraint $C_i$ is a triple $(W_i, S_i, P_i)$ of a positive weight $W_i > 0$ such that $\sum_{i=1}^m W_i = 1$, a subset of variables $S_i \subseteq \mathcal{X}$,  and a predicate $P_i: [R]^{S_i} \to \{0, 1\}$ that maps each assignment to variables in $S_i$ to $\{0, 1\}$. Here we use $[R]^{S_i}$ to denote the set of all functions from $S_i$ to $[R]$, i.e., $[R]^{S_i} = \{\psi : S_i \to [R]\}$.
\end{itemize}
For each assignment of variables $\varphi: \mathcal{X} \to [R]$, we define its value to be the total weights of the predicates satisfied by this assignment, i.e., $\sum_{i=1}^{m} W_i P_i(\varphi\mid_{S_i})$. The goal is to find an assignment $\varphi: \mathcal{X} \to [R]$ that with maximum value. We sometimes call the optimum the value of $(\mathcal{X}, \mathcal{C})$.
\end{definition}

Note that, while the standard definition of {\sc Max $k$-CSP$_R$} refers to the unweighted version where $W_1 = \cdots = W_m = 1/m$, Crescenzi, Silvestri and Trevisan showed that the approximability of these two cases are essentially the same~\cite{CST96}.\footnote{More specifically, they proved that, if the weighted version is NP-hard to approximate to within ratio $r$, then the unweighted version is also NP-hard to approximate to within $r - o_n(1)$ where $o_n(1)$ represents a function such that $o_n(1) \rightarrow 0$ as $n \rightarrow \infty$.} Hence, it is enough for us to consider just the weight version.

\subsection{Unique Games and $d$-to-1 Conjectures}

In this subsection, we give formal definitions for unique games, $d$-to-1 games and Khot's conjectures about them. First, we give a formal definition of unique games.
\begin{definition}[Unique Game]
A unique game $(V, W, E, N, \{\pi_e\}_{e \in E})$ consists of
\begin{itemize}
\item A bipartite graph $G = (V, W, E)$.
\item Alphabet size $N$.
\item For each edge $e \in E$, a permutation $\pi_e: [N] \to [N]$.
\end{itemize}
For an assignment $\varphi: V \cup W \to [N]$, an edge $e = (v, w)$ is satisfied
if $\pi_e(\varphi(v)) = \varphi(w)$. The goal is to find an assignment that satisfies as many edges as possible. We define the value of an instance to be the fraction of edges satisfied in the optimum solution.
\end{definition}

The Unique Games Conjecture states that it is NP-hard to distinguish an instance of value close one from that of value almost zero. More formally, it can be stated as follows.
\begin{conjecture}[Unique Games Conjecture]
For any sufficiently small $\eta, \gamma$, it is NP-hard to distinguish a unique game of value at least $1 - \eta$ from that of value at most $\gamma$.
\end{conjecture}

Next, we define $d$-to-1 games, which is similar to unique games except that each constraint is a $d$-to-1 function instead of a permutation.

\begin{definition}[$d$-to-1 Game]
A $d$-to-1 game $(V, W, E, N, \{\pi_e\}_{e \in E})$ consists of
\begin{itemize}
\item A bipartite graph $G = (V, W, E)$.
\item Alphabet size $N$.
\item For each edge $e \in E$, a function $\pi_e: [N] \to [N/d]$ such that $|\pi_e^{-1}(\sigma)| = d$ for every $\sigma \in [N/d]$.
\end{itemize}
Satisfiability of an edge, the goal, and an instance's value of is defined similar to that of unique games.
\end{definition}

In contrast to the UGC, $d$-to-1 Conjecture requires perfect completeness, i.e., it states that we cannot distinguish even a full satisfiable $d$-to-1 game from one with almost zero value.
\begin{conjecture}[$d$-to-1 Conjecture]
For any sufficiently small $\gamma$, it is NP-hard to distinguish a $d$-to-1 game of value 1 from that of value at most $\gamma$.
\end{conjecture}

\subsection{Fourier Expansions}
For any function $g: [q]^n \to \R$ over a finite alphabet $[q]$,
we define the Fourier expansion of $g$ as follows.

Consider the space of all functions $[q] \rightarrow \mathbb{R}$,
with the inner-product $\langle u, v \rangle := \E_{x \in [q]}[u(x)v(x)]$,
where the expectation is over a uniform $x \in [q]$.
Pick an orthonormal basis $l_1, \dots, l_{q}$ for this space
$l_i: \Sigma \rightarrow \mathbb{R}$, such that $l_1$ is the constant function $1$.
We can now write $g$ in the tensor-product basis, as
\begin{align}
g(x_1, x_2, \dots, x_n) = \sum_{s \in [q]^n} \hat{g}(s) \cdot \prod_{i=1}^n l_{s(i)}(x_i).
\end{align}
Since we pick $l_1(x) = 1$ for all $x \in [q]$, we also have
$\E_{x \in [q]}[l_i(x)] = \langle l_i, l_1 \rangle = 0$ for every $i \ne 1$.

Throughout, we use $\hat g$ to refer to the Fourier coefficients of a function
$g$.

For functions $g: [q]^n \to \R$, the $p$-norm is defined as
\beq
||g||_p = \E_{x \in [q]^n}[|g(x)|^p]^{1/p}.
\eeq
In particular, the squared $2$-norm is
\beq
||g||_2^2 = \E_{x \in [q]^n}[g(x)^2] =
\sum_{s \in [q]^n} \hat g(s)^2.
\eeq

\subsubsection{Noise Operators}
For $x \in [q]^n$, let $y \pcorr x$ denote that $y$ is a $\rho$-correlated copy of $x$.
That is, each coordinate $y_i$ is independently chosen to be equal to $x_i$ with
probability $\rho$, or chosen uniformly at random otherwise.

Define the noise operator $T_\rho$ acting on any function $g: [q]^n \to \R$ as
\beq
(T_\rho g)(x) = \E_{y \pcorr x}[g(y)].
\eeq
Notice that the noise operator $T_\rho$ acts on the Fourier coefficients on this basis as follows.
\begin{align}
f(x) = T_\rho g(x) = \sum_{s \in [q]^n} \hat{g}(s) \cdot \rho^{|s|} \cdot \prod_{i=1}^n l_{s(i)}(x_i)
\end{align}
where $|s|$ is defined as $|\{i \mid s(i) \ne 1\}|$.

\subsubsection{Degree Truncation}
For any function $g:[q]^n \to \R$, let $\gk$ denote the $(\leq d)$-degree part of $g$, i.e.,
\begin{align}
\gk(x) = \sum_{s \in [q]^n, |s| \leq d} \hat{g}(s) \cdot \prod_{i=1}^n l_{s(i)}(x_i),
\end{align}
and similarly let $g^{> d}: [q]^n \rightarrow \mathbb{R}$ denote the $(> k)$-degree part of $g$, i.e.,
\begin{align}
g^{> d}(x) = \sum_{s \in [q]^n, |s| > d} \hat{g}(s) \cdot \prod_{i=1}^n l_{s(i)}(x_i).
\end{align}

Notice that degree-truncation commutes with the noise-operator, so
writing $T_\rho \gk$ is unambiguous.

Also, notice that truncating does not increase 2-norm:
\beq
||\gk||_2^2
= \sum_{s \in [q]^n, |s| \leq d} \hat{g}(s)^2
\leq \sum_{s \in [q]^n} \hat{g}(s)^2
= ||g||_2^2.
\eeq

We frequently use the fact that noisy functions have small high-degree
contributions. That is, for any function $g: [q]^n \to [0, 1]$, we have
\bal
||T_\rho g^{> d}||_2^2
=\sum_{s \in [q]^n, |s| > d} \rho^{2|s|} \hat{g}(s) ^2
\leq \rho^{2d} \sum_{s \in [q]^n} \hat{g}(s)^2
= \rho^{2d} ||g||_2^2
\leq \rho^{2d}.
\eal

\subsubsection{Influences}
For any function $g:[q]^n \to \R$, the influence of coordinate $i \in [n]$ is
defined as
\beq
Inf_i[g] = \E_{x \in [q]^n}[ Var_{x_i \in [q]}[ g(x) ~|~ \{x_j\}_{j \neq i}] ].
\eeq
This can be expressed in terms of the Fourier coefficients of $g$ as follows:
\beq
Inf_i[g] = \sum_{s \in [q]^n:~s(i) \neq 1} \hat g(s)^2.
\eeq

Further, the degree-$d$ influences are defined as
\beq
Inf_i^{\leq d}[g]
= Inf_i[g^{\leq{d}}]
= \sum_{\substack{s \in [q]^n:\\ |s| \leq d, ~s(i) \neq 1}} \hat g(s)^2.
\eeq

\subsubsection{Binary Functions}
The previous discussion of Fourier analysis can be applied to boolean functions,
by specializing to $q=2$.
However, in this case the Fourier expansion can be written in a more convenient
form.
Let $G: \{+1, -1\}^n \to \R$ be a boolean function over $n$ bits.
We can choose orthonormal basis functions $l_1(y) = 1$ and $l_2(y) = y$, so
$G$ can be written as
\beq
G(y) = \sum_{S \subseteq [n]} \hat G(S) \prod_{i \in S} y_i
\eeq
for some coefficients $\hat G(S)$.

Degree-truncation then results in
\beq
G^{\leq d}(y) = \sum_{S \subseteq [n]: |S| \leq d} \hat G(S) \prod_{i \in S} y_i,
\eeq
and the noise-operator acts as follows:
\beq
T_\rho G(y) = \sum_{S \subseteq [n]} \hat G(S) \rho^{|S|} \prod_{i \in S} y_i.
\eeq

\subsubsection{Boolean Analogs}
To analyze $k$-CSP$_{R}$, we will want to translate between functions
over $[R]^n$ to functions over $\{\pm 1\}^{nR}$.
The following notion of \emph{boolean analogs} will be useful.

For any function $g: [R]^n \to \R$ with Fourier coefficients $\hat g(s)$,
define the boolean analog of $g$ to be a function $\analog{g}: \{\pm
1\}^{n \x R} \to \R$ such that
\beq
\analog{g}(z) = \sum_{s \in [R]^n} \hat{g}(s) \cdot \prod_{i \in [n], s(i) \ne 1} z_{i, s(i)}.
\eeq

Note that
\beq
||g||_2^2 =  \sum_{s \in [R]^n} \hat{g}(s)^2 = ||\analog{g}||_2^2,
\eeq
and that
\beq
\analog{g^{\leq d}} = \analog{g}^{\leq d}.
\eeq

Moreover, the noise operator acts nicely on $\analog{g}$ as follows:
\beq
T_\rho \analog{g} = \analog{T_\rho g}.
\eeq
For simplicity, we use $T_\rho$ to refer to both boolean and
non-boolean noise operators with correlation $\rho$.

\subsection{Invariance Principle and Mollification Lemma}

We start with the Invariance Principle in the form of
Theorem 3.18 in \cite{MOO10}:

\begin{theorem}[General Invariance Principle \cite{MOO10}]\label{thm:maininv}
    Let $f: [R]^n \to \R$ be any function such that $Var[f] \leq 1$,
    $Inf_i[f] \leq \delta$, and $deg(f) \leq d$.
    Let $F: \{\pm 1\}^{nR} \to \R$ be its boolean analog: $F = \analog{f}$.
    Consider any ``test-function'' $\psi: \R \to \R$ that is $\mathcal{C}^3$,
    with bounded 3rd-derivative $|\psi'''| \leq C$ everywhere.
    Then,
    \beq
    \left| \E_{y \in \{\pm 1\}^{nR}}[\psi(F(y))] - \E_{x \in [R]^n}[\psi(f(x))] \right|
    \leq C10^d R^{d/2} \sqrt{\delta}.
 \label{eqn:main_inv}
    \eeq
\end{theorem}

Note that the above version follows directly from Theorem 3.18 and Hypothesis 3 of \cite{MOO10},
and the fact that uniform $\pm 1$ bits are $(2, 3,
1/\sqrt{2})$-hypercontractive as described in \cite{MOO10}.

As we shall see later, we will want to apply the Invariance Principle
for some functions $\psi$ that are not in $\mathcal{C}^3$.
However, these functions will be Lipschitz-continuous with some constant $c \in \R$ (or
``$c$-Lipschitz''), meaning that
\beq
|\psi(x + \Delta) - \psi(x)| \leq c|\Delta| \quad \text{for all~} x,\Delta \in \R.
\eeq
Therefore, similar to Lemma 3.21 in~\cite{MOO10},
we can ``smooth'' it to get a function $\t \psi$ that is
that is $\mathcal{C}^3$, and has arbitrarily small pointwise difference to $\psi$.

\begin{lemma}[Mollification Lemma~\cite{MOO10}]
\label{lem:mollify}
Let $\psi: \R \to \R$ be any $c$-Lipschitz function.
Then for any $\zeta > 0$, there exists a function
$ \t \psi: \mathbb{R} \to \mathbb{R}$ such that
\begin{itemize}
\item $\t \psi \in \mathcal{C}^3$,
\item $||\t \psi - \psi||_{\infty} \leq \zeta$, and,
\item $||\t \psi'''||_{\infty} \leq \t C c^3/\zeta^2$.
\end{itemize}
For some universal constant $\t C$, not depending on $\zeta, c$.
\end{lemma}

For completeness, the full proof of the lemma can be found in Appendix~\ref{sec:apx-mollification}.

Now we state the following version of the Invariance Principle, which will be
more convenient to invoke. It can be proved simply by just combining the two previous lemmas. We list a full proof in Appendix~\ref{sec:app-ourinv}.

\begin{lemma}[Invariance Principle]
\label{ourinv}

Let $\psi: \R \to \R$ be one of the following functions:
\begin{enumerate}
    \item $\psi_1(t) := |t|$,
    \item $\psi_k(t) :=
\begin{cases}
t^k & \text{if } t \in [0, 1],\\
0 & \text{if } t < 0, \\
1 & \text{if } t \geq 1.
\end{cases}
$
\end{enumerate}

Let $f: [R]^n \to [0, 1]$ be any function with
all $Inf_i^{\leq d}[f] \leq \delta$.
Let $F: \{\pm 1\}^{nR} \to \R$ be its boolean analog: $F = \analog{f}$.
Let $\fk: [R]^n \to \R$ denote $f$ truncated to degree $d$, and similarly for
$\Fk: \{\pm 1\}^{nR} \to \R$.

Then, for parameters $d = 10k\log R$ and $\delta = 1/(R^{10 + 100 k \log(R)})$,
we have

\begin{equation}
\left| \E_{y \in \{\pm 1\}^{nR}}[\psi(\Fk(y))]
- \E_{x \in [R]^{n}}[\psi(\fk(x))] \right|
\leq O(1/R^k).
\end{equation}
\end{lemma}

\subsection{Hypercontractivity Theorem}

Another crucial ingredient in our proof is the hypercontractivity lemma, which says that, on $\{\pm 1\}^n$ domain, the operator $T_\rho$ smooths any function so well that the higher norm can be bound by the lower norm of the original (unsmoothed) function. Here we use the version of the theorem as stated in~\cite{OD14}.

\begin{theorem}[Hypercontractivity Theorem~\cite{OD14}]
Let $1 \leq p \leq q \leq \infty$. For any $\rho \leq \sqrt{\frac{p - 1}{q- 1}}$ and for any function $h: \{\pm 1\}^n \to \mathbb{R}$, the following inequality holds:
\begin{equation}
||T_\rho h||_q \leq ||h||_p.
\end{equation}
\end{theorem}

In particular, for choice of parameter
$\rho = 1/\sqrt{(k - 1)\log R}$, we have
\beq
||T_{2\rho} h||_k \leq ||h||_{1+\epsilon}.
\label{eq:ourhyper}
\eeq
where $\epsilon= 4/\log(R)$.

\section{Inapproximability of {\sc Max $k$-CSP$_R$}} \label{sec:inapprox}

In this section, we prove Theorem~\ref{thm:main} and Theorem~\ref{thm:d-to-1}. Before we do so, we first introduce a conjecture, which we name {\em One-Sided Unique Games Conjecture}. The conjecture is similar to UGC except that the completeness parameter $\zeta$ is fixed in contrast to UGC where the completeness can be any close to one.

\begin{conjecture}[One-Sided Unique Games Conjecture] \label{conj:OSUGC}
There exists $\zeta < 1$ such that, for any sufficiently small $\gamma$, it is NP-hard to distinguish a unique game of value at least $\zeta$ from that of value at most $\gamma$.
\end{conjecture}

It is obvious that the UGC implies One-Sided UGC with $\zeta = 1 - \eta$ for any sufficiently small $\eta$. It is also not hard to see that, by repeating each alphabet on the right $d$ times and spreading each $d$-to-1 constraint out to be a permutation, $d$-to-1 Games Conjecture implies One-Sided UGC with $\zeta = 1/d$. A full proof of this can be found in Appendix~\ref{sec:dto1osugc}.

Since both UGC and $d$-to-1 Games Conjecture imply One-Sided UGC, it is enough for us to show the following theorem, which implies both Theorem~\ref{thm:main} and Theorem~\ref{thm:d-to-1}.

\begin{theorem} \label{thm:osugmain}
Unless the One-Sided Unique Games Conjecture is false, for any $k \geq 2$ and any sufficiently large $R$, it is NP-hard to approximate {\sc {Max $k$-CSP}$_R$} to within $2^{O(k \log k)}(\log R)^{k/2}/R^{k - 1}$ factor.
\end{theorem}

The theorem will be proved in Subsection~\ref{subsec:red-ug}. Before that, we first prove an inequality that is the heart of our soundness analysis in Subsection~\ref{subsec:noisestability}.

\subsection{Parameters} \label{subsec:params}

We use the following parameters throughout, which we list for convenience here:
\begin{itemize}
    \item Correlation: $\rho = 1/\sqrt{(k - 1)\log R}$
    \item Degree: $d = 10k\log R$
    \item Low-degree influences: $\delta = 1/(R^{10 + 100 k \log(R)})$
\end{itemize}

\subsection{Hypercontractivity for Noisy Low-Influence Functions} \label{subsec:noisestability}
Here we show a version of hypercontractivity for noisy low-influence functions
over large domains.
Although hypercontractivity does not hold in general for
noisy functions over large domains, it turns out to hold in our setting
of high-noise and low-influences.
The main technical idea is to use the Invariance Principle
to reduce functions over larger domains to boolean functions,
then apply boolean hypercontractivity.

\begin{lemma}[Main Lemma]
\label{mainlemma}
Let $g: [R]^n \to [0, 1]$ be any function with $\E_{x \in [R]^n}[g(x)] = 1/R$.
Then, for our choice of parameters $\rho, d, \delta$:
If $Inf_i^{\leq d}[g] \leq \delta$ for all $i$,
then
$$\E_{x \in [R]^n}[(T_\rho g(x))^k] \leq 2^{O(k)}/R^k.$$
\end{lemma}

Before we present the full proof, we outline the high-level steps below.
Let $f = T_\rho g$, and define boolean analogs $G = \analog{g}$, and $F = \analog{f}$.
Let $\psi_k: \R \to \R$ be defined as in Lemma~\ref{ourinv}. Then,

\begin{align}
\E_{x \in [R]^n}[f(x)^k]
&\approx \E[\psi_k(\fk(x))] \label{eq:step1}\\
(\text{Lemma~\ref{ourinv}: Invariance Principle})\quad &\approx \E_{y \in \{\pm 1\}^{nR}}[\psi_k(\Fk(y))]
\label{eq:step2} \\
(\text{Definition of } \psi_k)\quad &\leq ||\Fk||_k^k \\
(\text{Definition of } F)\quad &= ||T_\rho \Gk||_k^k \\
&= ||T_{2\rho}T_{1/2} \Gk||_k^k\\
(\text{Hypercontractivity, for small $\epsilon$})\quad
&\leq ||T_{1/2}\Gk||_{1+\epsilon}^k \label{eq:stephyper}\\
(\text{Invariance, etc.})\quad
&\approx 2^{O(k)} || g ||_{1}^k \\
(\text{Since } \E[|g|] = 1/R)\quad &= 2^{O(k)}/R^k. \label{eq:eg}
\end{align}

\proof

To establish line ~\eqref{eq:step1}, first notice that
\beq
\psi_k(f(x)) = \psi_k(\fk(x) + f^{> d}(x)) \leq \psi_k(\fk(x)) + k|f^{> d}(x)|
\eeq
where the last inequality is because the function $\psi_k$ is $k$-Lipschitz.

Moreover, since $g(x) \in [0, 1]$, we have $f(x) \in [0, 1]$, so
\beq
f(x)^k = \psi_k(f(x)).
\eeq

Thus,
\bal
\E[ f(x)^k ] &= \E[\psi_k(f(x))] \\
&\leq \E [\psi_k(\fk(x))] + k\E[ |f^{> d}(x)| ]\\
&= \E [\psi_k(\fk(x))] + k||f^{> d}||_1\\
&\leq \E [\psi_k(\fk(x))] + k||f^{> d}||_2.\\
\eal
And we can bound the 2-norm of $f^{> d}$, since $f$ is noisy, we have
\bal
||f^{> d}||_2^2
= ||T_\rho g^{> d}||_2^2
\leq \rho^{2d}
\leq O(1/R^{2k}).
\eal
The last inequality comes from our choice of $\rho$ and $d$.

So line ~\eqref{eq:step1} is established:
\beq
\E[f(x)^k] \leq \E[\psi_k(\fk(x))] + O(k/R^k).
\eeq

Line ~\eqref{eq:step2}
follows directly from our version of the Invariance Principle (Lemma~\ref{ourinv}),
for the function $\psi_k$:
\beq
\E_{x \in [R]^n}[\psi_k(\fk(x))] \leq \E_{y \in \{\pm 1\}^{nR}}[\psi_k(\Fk(y))] + O(1/R^k).
\eeq

We can now rewrite $\E_{y \in \{\pm 1\}^{nR}}[\psi_k(\Fk(y))]$ as
\begin{align}
    \E_{y \in \{\pm 1\}^{nR}}[\psi_k(\Fk(y))] &\leq \E_{y \in \{\pm 1\}^{nR}}[|\Fk(y)|^k] \\
    &= ||\Fk||_k^k \\
    &= ||T_\rho \Gk||_k^k \\
    &= ||T_{2 \rho}T_{1/2} \Gk||_k^k. \\
\end{align}

Now, from the Hypercontractivity Theorem, Equation \eqref{eq:ourhyper}, we have
\beq
||T_{2\rho}T_{1/2} \Gk||_k \leq ||T_{1/2}\Gk||_{1+\epsilon}
\label{eq:hyper_use}
\eeq
for $\epsilon = 4/\log{R}$.
This establishes line~\eqref{eq:stephyper}:
\bal
||T_{2 \rho}T_{1/2} \Gk||_k^k
\leq
||T_{1/2}\Gk||_{1+\epsilon}^k
&= \E[|T_{1/2} \Gk(y)|^{1+\epsilon}]^{k/(1+\epsilon)}.
\label{eq:hyperproof}
\eal

To show the remaining steps, we will apply the Invariance Principle once more.
Notice that for all $t \in \R: |t|^{1+\epsilon} \leq |t| + t^2$. Hence, we can
derive the following bound:
\bal
\E[|T_{1/2} \Gk(y)|^{1+\epsilon}]
&\leq
\E[|T_{1/2} \Gk(y)|] + \E[(T_{1/2} \Gk(y))^2]\\
(\text{Matching Fourier expansion})\quad
&= \E[|T_{1/2} \Gk(y)|] + \E[(T_{1/2} \gk(y))^2]\\
(\text{Lemma~\ref{ourinv}, Invariance Principle})\quad
&\leq \E[|T_{1/2} \gk(x)|] + \E[(T_{1/2} \gk(x))^2] + O(1/R^k).
\label{eq:invar2}
\eal
Here we applied our Invariance Principle (Lemma~\ref{ourinv}) for the
function $\psi_1$ as defined in Lemma~\ref{ourinv}.
We will bound each of the expectations on the RHS, using the fact that $g$ is
balanced, and $T_{1/2} g$ is noisy.

First,
\bal
\E[|T_{1/2} \gk(x)|]
&=\E[|T_{1/2} g(x) - T_{1/2}g^{> d}(x)|]\\
(\text{Triangle Inequality})\quad
&\leq \E[|T_{1/2} g(x)|] + \E[|T_{1/2}g^{> d}(x)|]\\
&= ||g||_1 + ||T_{1/2} g^{>d}||_1\\
&\leq ||g||_1 + ||T_{1/2} g^{>d}||_2\\
&\leq 1/R + (1/2)^d \\
(\text{By our choice of } d)\quad
&= O(1/R).
\eal

Second,
\bal
\E[(T_{1/2} \gk(x))^2]
&= \sum_{s \in [R]^n, |s| \leq d} (1/2)^{2|s|} \hat{g}(s)^2\\
&\leq \sum_{s \in [R]^n} (1/2)^{2|s|} \hat{g}(s)^2\\
&= \E[(T_{1/2} g(x))^2]\\
(\text{Since } g \in [0, 1])\quad
&\leq \E[T_{1/2} g(x)]\\
&= \E[g(x)] = 1/R.
\eal

Finally, plugging these bounds into
\eqref{eq:invar2},
we find:
\bal
||T_{1/2}\Gk||_{1+\epsilon}^k
&= \E[|T_{1/2} \Gk(y)|^{1+\epsilon}]^{k/(1+\epsilon)}\\
&\leq ( O(1/R) )^{k/(1+\epsilon)}\\
&= 2^{O(k)}/R^{k/(1+\epsilon)}\\
&\leq 2^{O(k)}/R^{k(1-\epsilon)}\\
(\text{Recall } \epsilon=4/\log R)\quad
&= 2^{O(k)}/R^{k}.
\eal
This completes the proof of the main lemma. \qed

\subsection{Reducing Unique Label Cover to {\sc Max $k$-CSP$_R$}} \label{subsec:red-ug}
Here we reduce unique games to {\sc Max $k$-CSP$_R$}.
We will construct a PCP verifier that reads $k$ symbols of the proof
(with an alphabet of size $R$) with the following properties:
\begin{itemize}
\item {\bf (Completeness)} If the unique game has value at least $\zeta$, then the verifier accepts an honest proof with probability at least $c = 1/((\log R)^{k/2}2^{O(k \log k)})$.
\item {\bf (Soundness)} If the unique game has value at most $\gamma = 2^{O(k)}\delta^2/(4d R^k)$, then the verifier accepts any (potentially cheating) proof with probability at most $s = 2^{O(k)}/R^{k - 1}$.
\end{itemize}
Since each symbol in the proof can be viewed as a variable and each accepting predicate of the verifier can be viewed as a constraint of {\sc Max $k$-CSP}$_R$, assuming the One-sided UGC, this PCP implies NP-hardness of approximating {\sc Max $k$-CSP}$_R$ of factor $s/c = 2^{O(k \log k)}(\log R)^{k/2}/R^{k-1}$ and, hence, establishes our Theorem~\ref{thm:osugmain}.

\subsubsection{The PCP}
Given a unique game $(V, W, E, n, \{\pi_e\}_{e \in E})$,
the proof is the truth-table of a function $h_{w}: [R]^n \to [R]$ for each vertex $w \in W$.
By folding, we can assume $h_w$ is balanced, i.e.
$h_w$ takes on all elements of its range with equal probability:
$\Pr_{x \in [R]^n}[h_w(x) = i] = 1/R$ for all $i \in [R]$.
\footnote{
    More precisely, if the truth-table provided is of some function $\t h_w:
    [R]^n \to [R]$, we define the ``folded'' function $h_w$ as
    $h_w(x_1, x_2, x_3, \dots x_n) := \t h_w(x - (x_1, x_1, \dots, x_1)) + x_1$,
    where the $\pm$ is over mod $R$.
    Notice that the folded $h_w$ is balanced, and also that folding does not
    affect dictator functions.
    Thus we define our PCP in terms of $h_w$, but simulate queries to $h_w$
    using the actual proof $\t h_w$.
}

Notationally, for $x \in [R]^n$, let $(x \circ \pi)$ denote permuting the
coordinates of $x$ as: $(x \circ \pi)_i = x_{\pi(i)}$.
Also, for an edge $e = (v, w)$, we write $\pi_e = \pi_{v, w}$, and define
$\pi_{w, v} = \pi_{v, w}^{-1}$.

The verifier picks a uniformly random vertex $v \in V$,
and $k$ independent uniformly random neighbors of $v$: $w_1, w_2, \dots, w_k \in W$.
Then pick $z \in [R]^n$ uniformly at random,
and let $\i{x}{1}, \i{x}{2}, \dots, \i{x}{k}$ be independent $\rho$-correlated noisy copies of $z$
(each coordinate $x_i$ chosen as equal to $z_i$ w.p. $\rho$, or uniformly at random otherwise).
The verifier accepts if and only if
\begin{equation}
h_{w_1}(\i{x}{1} \circ \pi_{w_1, v})
=
h_{w_2}(\i{x}{2} \circ \pi_{w_2, v})
=
\cdots
=
h_{w_k}(\i{x}{k} \circ \pi_{w_k, v}).
\end{equation}

To achieve the desired hardness result, we pick $\rho = 1/\sqrt{(k - 1)\log R}$.

\subsubsection{Completeness Analysis}

First, note that that we can assume without loss of generality that the graph is regular on $V$ side.\footnote{See, for instance, Lemma 3.4 in~\cite{KR08}.} Let the degree of each vertex in $V$ be $d$.

Suppose that the original unique game has an assignment of value at least
$\zeta$. Let us call this assignment $\varphi$.
The honest proof defines $h_w$ at each vertex $w \in W$ as the long code
encoding of this assignment, i.e., $h_w(x) = x_{\varphi(w)}$.
We can written the verifier acceptance condition as follows:
\begin{align}
\text{The verifier accepts} &\Leftrightarrow
h_{w_1}(\i{x}{1} \circ \pi_{w_1, v})
=
\cdots
=
h_{w_k}(\i{x}{k} \circ \pi_{w_k, v})
 \\
&\Leftrightarrow
(\i{x}{1} \circ \pi_{w_1, v})_{\varphi(w_1)}
=
\cdots
=
(\i{x}{k} \circ \pi_{w_k, v})_{\varphi(w_k)}
 \\
&\Leftrightarrow
(\i{x}{1})_{\pi_{w_1, v}(\varphi(w_1))}
=
\cdots
=
(\i{x}{k})_{\pi_{w_k, v}(\varphi(w_k))}.
\end{align}

Observe that, if the edges $(v, w_1), \dots, (v, w_k)$ are satisfied by $\varphi$,
then $\pi_{w_1, v}(\varphi(w_1)) = \dots = \pi_{w_k, v}(\varphi(w_k)) = \varphi(v)$.
Hence, if the aforementioned edges are satisfied and $\i{x}{1}, \dots, \i{x}{k}$
are not perturbed at coordinate $\varphi(v)$, then
$(\i{x}{1})_{\pi_{w_1, v}(\varphi(w_1))} = \cdots = (\i{x}{k})_{\pi_{w_k, v}(\varphi(w_k))}$.

For each $u \in V$, let $s_u$ be the number of satisfied edges touching $u$. Since $w_1, \dots, w_k$ are chosen from the neighbors of $v$ independently from each other, the probability that the edges $(v, w_1), (v, w_2), \dots, (v, w_k)$ are satisfied can be bounded as follows:
\begin{align}
&\Pr_{v, w_1, \dots, w_k}[(v, w_1), \dots, (v, w_k) \text{ are satisfied}] \\
&= \sum_{u \in V}\Pr_{w_1, \dots, w_k}[(v, w_1), \dots, (v, w_k) \text{ are satisfied} \mid v = u]\Pr[v = u] \\
&= \sum_{u \in V} \left(s_u/d\right)^k \Pr[v = u] \\
&= \E_{u \in V}\left[\left(s_u/d\right)^k\right] \\
&\geq \E_{u \in V}\left[s_u/d\right]^k.
\end{align}

Notice that $\E_{u \in V}\left[s_u/d\right]$ is exactly the value of $\varphi$, which is at least $\zeta$. As a result, $$\Pr_{v, w_1, \dots, w_k}[(v, w_1), \dots, (v, w_k) \text{ are satisfied}] \geq \zeta^k.$$

Furthermore, it is obvious that the probability that $x_1, \dots, x_k$ are not perturbed at the coordinate $\varphi(v)$ is $\rho^k$. As a result, the PCP accepts with probability at least $\zeta^k \rho^k$. When $\rho = 1/\sqrt{(k - 1)\log R}$ and $\zeta$ is a constant not depending on $k$ and $R$, the completeness is $1/((\log R)^{k/2}2^{O(k \log k)})$.

\subsubsection{Soundness Analysis}
Suppose that the unique game has value at most $\gamma = 2^{O(k)}\delta^2/(4d R^k)$.
We will show that the soundness is $2^{O(k)}/R^{k - 1}$.

Suppose for the sake of contradiction that the probability that the verifier accepts
$\Pr[accept] > t = 2^{\Omega(k)}/R^{k - 1}$
where $\Omega(\cdot)$ hides some large enough constant.

Let $h_w^i(x) : [R]^n \to \{0, 1\}$ be the indicator function for $h_w(x) = i$ and
let $x \pcorr z$ denote that $x$ is a $\rho$-correlated copy of $z$. We have
\begin{align}
\Pr[accept] &= \Pr[
h_{w_1}(\i{x}{1} \circ \pi_{w_1, v})
=
\cdots
=
h_{w_k}(\i{x}{k} \circ \pi_{w_k, v})
] \\
&= \sum_{i \in [R]}
\Pr[
i = h_{w_1}(\i{x}{1} \circ \pi_{w_1, v})
=
\cdots
=
h_{w_k}(\i{x}{k} \circ \pi_{w_k, v})
] \\
&= \sum_{i \in [R]}
\E[
h^i_{w_1}(\i{x}{1} \circ \pi_{w_1, v})
\cdots
h^i_{w_k}(\i{x}{k} \circ \pi_{w_k, v})
] \\
(\text{Since } w_i \text{'s are independent given } v)
&= \sum_{i \in [R]}
\E\left[
\E_{w_1}[
h^i_{w_1}(\i{x}{1} \circ \pi_{w_1, v})
]
\cdots
\E_{w_k}[
h^i_{w_k}(\i{x}{k} \circ \pi_{w_k, v})
]
\right].
\end{align}
Now define $g^i_v : [R]^n \to [0, 1]$ as
\begin{equation}
g^i_v(x) = \E_{w \sim v}[h^i_{w}(x \circ \pi_{w, v})]
\end{equation}
where $w \sim v$ denotes neighbors $w$ of $v$.

We can rewrite $\Pr[accept]$ as follows:
\begin{align}
\Pr[accept]
&= \sum_{i \in [R]}
\E[ g^i_v(\i{x}{1}) g^i_v(\i{x}{2}) \cdots g^i_v(\i{x}{k}) ] \\
(\text{Since } \i{x}{j} \text{'s are independent given } z)
&= \sum_{i \in [R]}
\E\left[ \E_{x \pcorr z}[ g^i_v(x) ] ^k \right] \\
&= \sum_{i \in [R]}\E_{v, z}[ (T_\rho g^i_v(z))^k ]\\
&= \E_v\left[ \sum_{i \in [R]}\E_{z}[ (T_\rho g^i_v(z))^k \right]. \label{eqn:evz}
\end{align}

Next, notice that
\begin{align}
\sum_{i \in [R]}\E_{z}[ (T_\rho g^i_v(z))^k ]
&=
\E_z\left[ \sum_{i \in [R]} (T_\rho g^i_v(z))^k \right] \\
&\leq
\E_z\left[ \left(\sum_{i \in [R]} T_\rho g^i_v(z)\right)^k \right] \\
&=
\E_z\left[ \left(T_\rho \sum_{i \in [R]} g^i_v(z)\right)^k \right] \\
&=
\E_z[ (T_\rho 1)^k ] = 1.
\end{align}

Therefore, since $\Pr[accept] > t$,
 by~\eqref{eqn:evz}, at least $t/2$ fraction of vertices $v \in V$ have
\begin{equation}
\sum_{i \in [R]}\E_{z}[ (T_\rho g^i_v(z))^k ] \geq t/2.
\end{equation}

For these ``good'' vertices, there must exist some $i \in [R]$ for which
\begin{equation}
\E_{z}[ (T_\rho g^{i}_v(z))^k ] \geq t/(2R).
\end{equation}

Then for ``good'' $v$ and $i$ as above,
\begin{equation}
\E_{z}[ (T_\rho g^{i}_v(z))^k ] > 2^{\Omega(k)}/R^k.
\end{equation}

By Lemma~\ref{mainlemma} (Main Lemma), this means $g_v^{i}$ has some coordinate $j$
for which
\begin{equation}
Inf_j^{\leq d}[g^{i}_v] > \delta
\end{equation}
for our choice of $d, \delta$ as defined in Subsection~\ref{subsec:params}.
Pick this $j$ as the label of vertex $v \in V$.

Now to pick the label of a vertex $w \in W$, define
the candidate labels as
\beq
Cand[w] = \{j \in [n] : \exists~i \in [R] \text{ s.t. }
Inf_j^{\leq d}[h^i_w] \geq \delta/2 \}.
\eeq
Notice that
\beq
\sum_{j \in [n]} Inf_j^{\leq d}[h^i_w]
= \sum_{s \in [R]^n: ~|s| \leq d} |s| \hat h^i_w(s)^2
\leq d \sum_{s: |s| > 0}\hat h^i_w(s)^2
= d ~Var[h^i_w] \leq d.
\eeq

So for each $i \in [R]$, the projection $h_w^i$ can have at most $2d/\delta$ coordinates
with influence $\geq \delta/2$.
Therefore the number of candidate labels is bounded:
\beq
|Cand[w]| \leq {2dR}/\delta.
\eeq

Now we argue that picking a random label in $Cand[w]$ for $w \in W$
is in expectation a good decoding.
We will show that if we assigned label $j$ to a ``good'' $v \in V$,
then $\pi_{v,w}(j) \in Cand[w]$ for a constant fraction of neighbors
$w \sim v$. Note here that $\pi_{v,w} = \pi^{-1}_{w,v}$.

First, since $g^i_v(x) = \E_{w \sim v}[h^i_{w}(x \circ \pi_{w, v})]$, the Fourier transform of $g_v^i$ is related to the Fourier transform of the long code labels $h_w^i$ as
\begin{align}
\hat g_v^i(s) &= \E_{w \sim v}[\hat h^i_{w}(s \circ \pi_{w, v})].
\end{align}

Hence, the influence $Inf_j^{\leq d}[g^i_v]$ of being large implies
the expected influence $In\fk_{\pi_{v,w}^{-1}(j)}[h_w^i]$
of its neighbor labels $w \sim v$ is also large as formalized below.

\begin{align}
\delta < Inf_j^{\leq k}[g^{i}_v]
&= \sum_{\substack{s \in [R]^n \\ |s| \leq k, s_j \neq 1}} \hat g^{i}_v(s)^2 \\
&= \sum \E_{w \sim v}[\hat h^{i}_{w} (s \circ \pi_{w,v})]^2\\
&\leq \sum \E_{w \sim v}[\hat h^{i}_{w} (s \circ \pi_{w,v})^2] \\
&= \E_{w \sim v}[\sum_{\substack{s \in [R]^n \\ |s| \leq k, s_j \neq 1}}
\hat h^{i}_{w} (s \circ \pi_{w,v})^2]\\
&= \E_{w \sim v}[\sum_{\substack{s \in [R]^n \\ |s| \leq k, s_{\pi_{w,v}^{-1}(j)} \neq 1}}
\hat h^{i}_{w} (s)^2]\\
(\text{Since } \pi_{v,w} = \pi^{-1}_{w,v})
&= \E_{w \sim v}[\sum_{\substack{s \in [R]^n \\ |s| \leq k, s_{\pi_{v,w}(j)} \neq 1}}
\hat h^{i}_{w} (s)^2]\\
&= \E_{w \sim v}[In\fk_{\pi_{v,w}(j)}[h_w^i]]
\end{align}

Therefore, at least $\delta / 2$ fraction of neighbors $w \sim v$
must have
$
In\fk_{\pi_{v,w}(j)}[h_w^i] \geq \delta / 2
$, and so $\pi_{v, w}(j) \in Cand[w]$ for at least $\delta / 2$ fraction of
neighbors of ``good'' vertices $v$.

Finally, recall that at least $(t/2)$ fraction of vertices $v \in V$
are ``good''. These vertices have at least $(\delta / 2)$
fraction of neighbors $w \in W$ with high-influence labels
and the matching label $w \in W$ is picked with probability at least $\delta/(2dR)$.
Moreover, as stated earlier, we can assume that the graph is regular on $V$ side.
Hence, the expected fraction of edges satisfied by this decoding is
at least
\beq
(t/2)(\delta/2) (\delta/2dR) = t\delta^2/(4dR) = 2^{\Omega(k)}\delta^2/(4d R^k) > \gamma,
\eeq
which contradicts our assumption that the unique game has value at most $\gamma$. Hence, we can conclude that the soundness is at most $2^{O(k)}/R^{k - 1}$ as desired.

\section{$\Omega(\log R/R^{k - 1})$-Approximation Algorithm for {\sc Max $k$-CSP$_R$}} \label{sec:approx}

Instead of just extending the KKT algorithm to work with {\sc Max $k$-CSP}s, we will show a more generalized statement that \emph{any} algorithm that approximates {\sc Max CSP}s with small arity can be extended to approximate {\sc Max CSP}s with larger arities. In particular, we show how to extend any $f(R)/R^{k'}$-approximation algorithm for {\sc Max $k'$-CSP$_R$} to an $(f(R)/2^{O(\min\{k', k - k'\})})/R^{k}$-approximation algorithm for {\sc Max $k$-CSP$_R$} where $k > k'$.

Since the naive algorithm that assigns every variable randomly has an approximation ratio of $1/R^{k}$, we think of $f(R)$ as the advantage of algorithm $A$ over the randomized algorithm. From this perspective, our extension lemma preserves the advantage up to a factor of $1/2^{O(\min\{k', k - k'\})}$.

The extension lemma and its proof are stated formally below.

\begin{lemma} \label{lem:ext}
Suppose that there exists a polynomial-time approximation algorithm $A$ for {\sc Max $k'$-CSP$_R$} that outputs an assignment with expected value at least $f(R)/R^{k'}$ times the optimum. For any $k > k'$, we can construct a polynomial-time approximation algorithm $B$ for {\sc Max $k$-CSP$_R$} that outputs an assignment with expected value at least $(f(R)/2^{O(\min\{k', k - k'\})})/R^{k}$ times the optimum.
\end{lemma}

\begin{proof}
The main idea of the proof is simple. We turn an instance of {\sc Max $k$-CSP$_R$} to an instance of {\sc Max $k'$-CSP$_R$} by constructing $\binom{k}{k'}R^{k - k'}$ new constraints for each original constraint; each new constraint is a projection of the original constraint to a subset of variables of size $k'$. We then use $A$ to solve the newly constructed instance. Finally, $B$ simply assigns each variable with the assignment from $A$ with a certain probability and assign it randomly otherwise.

For convenience, let $\alpha$ be $\frac{k - k'}{k}$. We define $B$ on input $(\mathcal{X}, \mathcal{C})$ as follows:
\begin{enumerate}
\item Create an instance $(\mathcal{X}, \mathcal{C}')$ of {\sc Max $k'$-CSP$_R$} with the same variables and, for each $C = (W, S, P) \in \mathcal{C}$ and for every subset $S'$ of $S$ with $|S'| = k'$ and every $\tau \in [R]^{S - S'}$, create a constraint $C^{S', \tau} = (W', S', P')$ in $\mathcal{C}'$ where $W'= \frac{W}{\binom{k}{k'}R^{k - k'}}$ and $P': [R]^{S'} \rightarrow \{0, 1\}$ is defined by $$P'(\psi) = P(\psi \circ \tau).$$ Here $\psi \circ \tau$ is defined as follows:
\begin{align*}
\psi \circ \tau(x) =
\begin{cases}
\psi(x) & \text{if } x \in S', \\
\tau(x) & \text{otherwise.}
\end{cases}
\end{align*}
\item Run $A$ on input $(\mathcal{X}, \mathcal{C}')$. Denote the output of $A$ by $\varphi_A$.
\item \label{step:rand-asg} For each $x \in \mathcal{X}$, with probability $\alpha$, pick $\varphi_B(x)$ randomly from $[R]$. Otherwise, let $\varphi_B(x)$ be $\varphi_A(x)$.
\item Output $\varphi_B$.
\end{enumerate}

We now show that $\varphi_B$ has expected value at least $(f(R)/2^{O(\min\{k', k - k'\})})/R^{k}$ times the optimum.

First, observe that the optimum of $(\mathcal{X}, \mathcal{C}')$ is at least $1/R^{k - k'}$ times that of $(\mathcal{X}, \mathcal{C})$. To see that this is true, consider any assignment $\varphi: \mathcal{X} \to [R]$ and any constraint $C = (W, S, P)$. Its weighted contribution in $(\mathcal{X}, \mathcal{C})$ is $W P(\varphi|_S)$. On the other hand, $\frac{W}{\binom{k}{k'}R^{k - k'}} P(\varphi|_S)$ appears $\binom{k}{k'}$ times in $(\mathcal{X}, \mathcal{C}')$, once for each subset $S' \subseteq S$ of size $k'$. Hence, the value of $\varphi$ with respect to $(\mathcal{X}, \mathcal{C}')$ is at least $1/R^{k - k'}$ times its value with respect to $(\mathcal{X}, \mathcal{C})$

Recall that the algorithm $A$ gives an assignment of expected value at least $f(R)/R^{k'}$ times the optimum of $(\mathcal{X}, \mathcal{C}')$. Hence, the expected value of $\varphi_A$ is at least $f(R)/R^{k}$ times the optimum of $(\mathcal{X}, \mathcal{C})$.

Next, we will compute the expected value of $\varphi_B$ (with respect to $(\mathcal{X}, \mathcal{C})$). We start by computing the expected value of $\varphi_B$ with respect to a fixed constraint $C = (W, S, P) \in \mathcal{C}$, i.e., $\E_{\varphi_B}[W P(\varphi_B|_S)]$. For each $S' \subseteq S$ of size $k$, we define $D_{S'}$ as the event where, in step~\ref{step:rand-asg}, $\varphi_B(x)$ is assigned to be $\varphi_A(x)$ for all $x \in S'$ and $\varphi_B(x)$ is assigned randomly for all $x \in S - S'$.

Since $D_{S'}$ is disjoint for all $S' \subseteq S$ of size $k$, we have the following inequality.
\begin{align}
\E_{\varphi_B}[W P(\varphi_B|_S)] &\geq \sum_{{S' \subseteq S \atop |S'| = k}} \Pr[D_{S'}] \E_{\varphi_B}[W P(\varphi_B|_S) \mid D_{S'}] \\
(\text{Since } \Pr[D_{S'}] = \alpha^{k - k'}(1 - \alpha)^{k'}) &= \alpha^{k - k'}(1 - \alpha)^{k'} \sum_{{S' \subseteq S \atop |S'| = k}} W \E_{\varphi_B}[P(\varphi_B|_S) \mid D_{S'}]
\end{align}

 Moreover, since every vertex in $S - S'$ is randomly assigned when $D_{S'}$ occurs, $\E[P(\varphi_B|_S) \mid D_{S'}]$ can be view as the average value of $P((\varphi_A|_{S'}) \circ \tau)$ over all $\tau\in [R]^{S - S'}$. Hence, we can derive the following:
\begin{align}
\E_{\varphi_B}[P(\varphi_B|_S) \mid D_{S'}] = \frac{1}{R^{k - k'}} \E_{\varphi_A}\left[\sum_{\tau\in [R]^{S - S'}} P((\varphi_A|_{S'}) \circ \tau)\right].
\end{align}

As a result, we have
\begin{align} \label{inq:one-cons}
\E_{\varphi_B}[W P(\varphi_B|_S)] \geq \frac{\alpha^{k - k'}(1 - \alpha)^{k'}}{R^{k - k'}} \left( \E_{\varphi_A}\left[\sum_{{S' \subseteq S \atop |S'| = k}} \sum_{\tau\in [R]^{S - S'}} W P((\varphi_A|_{S'}) \circ \tau) \right]\right).
\end{align}

By summing (\ref{inq:one-cons}) over all constraints $C \in \mathcal{C}$, we arrive at the following inequality.
\begin{align}
&\E_{\varphi_B}\left[\sum_{C = (W, S, P) \in \mathcal{C}} W P(\varphi_B|_S)\right]  \\
&\geq \frac{\alpha^{k - k'}(1 - \alpha)^{k'}}{R^{k - k'}} \E_{\varphi_A}\left[ \sum_{C = (W, S, P) \in \mathcal{C}} \left( \sum_{{S' \subseteq S \atop |S'| = k}} \sum_{\tau\in [R]^{S - S'}} W P((\varphi_A|_{S'}) \circ \tau) \right)\right] \\
&= \binom{k}{k'}\alpha^{k - k'}(1 - \alpha)^{k'} \E_{\varphi_A}\left[ \sum_{C = (W, S, P) \in \mathcal{C}} \left( \sum_{{S' \subseteq S \atop |S'| = k}} \sum_{\tau\in [R]^{S - S'}} \frac{W}{\binom{k}{k'}R^{k - k'}} P((\varphi_A|_{S'}) \circ \tau) \right)\right] \\
&= \binom{k}{k'}\alpha^{k - k'}(1 - \alpha)^{k'} \E_{\varphi_A}\left[\sum_{C' = (W', S', P') \in \mathcal{C}} W' P'(\varphi_A|_{S'})\right]
\end{align}

The first expression is the expected value of $\varphi_B$ whereas the last is $\binom{k}{k'}\alpha^{k - k'}(1 - \alpha)^{k'}$ times the expected value of $\varphi_A$. Since the expected value of $\varphi_A$ is at least $f(R)/R^{k}$ times the optimum of $(\mathcal{X}, \mathcal{C})$, the expected value of $\varphi_B$ is at least $(\binom{k}{k'}\alpha^{k - k'}(1 - \alpha)^{k'})(f(R)/R^{k})$ times the optimum of $(\mathcal{X}, \mathcal{C})$.

Finally, we substitute $\alpha = \frac{k - k'}{k}$ in to get
\begin{align}
\binom{k}{k'}\alpha^{k - k'}(1 - \alpha)^{k'} &= \binom{k}{k'}\left(\frac{k - k'}{k}\right)^{k - k'}\left(\frac{k'}{k}\right)^{k'}.
\end{align}

Let $l = \min\{k', k - k'\}$. We then have
\begin{align}
\binom{k}{k'}\left(\frac{k - k'}{k}\right)^{k - k'}\left(\frac{k'}{k}\right)^{k'} &= \binom{k}{l}\left(\frac{k - l}{k}\right)^{k - l}\left(\frac{l}{k}\right)^{l} \\
&\geq \left(\frac{k}{l}\right)^{l} \left(\frac{k - l}{k}\right)^{k - l}\left(\frac{l}{k}\right)^{l} \\
&\geq \left(\frac{k - l}{k}\right)^{k} \\
&= \left(\left(1 - {l}/k\right)^{2k/{l}}\right)^{2 l} \\
(\text{From Bernoulli's inequality and from } l \leq k/2)
&\geq 1/2^{2l}.
\end{align}

Hence, $\varphi_B$ has expected value at least $(f(R)/2^{O(l)})/R^{k}$ times the optimum of $(\mathcal{X}, \mathcal{C})$, which completes the proof of this lemma.
\qed
\end{proof}

Finally, Theorem~\ref{thm:approx} is an immediate consequence of applying Lemma~\ref{lem:ext} to the algorithm from~\cite{KKT15} with $k' = 2$ and $f(R) = \Omega(R \log R)$.

\section{$k$-Query Large Alphabet Dictator Test}
\label{sec:dtest}
We remark that the results of Section~\ref{sec:inapprox} also implicitly yield
a $k$-query nonadaptive \emph{Dictator-vs.-Quasirandom} test for
functions over large alphabet.
A Dictator-vs.-Quasirandom test aims to distinguish dictator functions from
functions with small low-degree influences (``quasirandom'').
This concept was essentially introduced in \cite{hastad1996},
and we borrow the ``quasirandom'' terminology from \cite{3bit}
(adapted here for functions over non-binary alphabets).
Specifically, we have the following test:

\begin{theorem}
\label{thm:dtest}
For any function $f: [R]^n \to [R]$, and any $i \in [R]$,
let $f^i: [R]^n \to \{0, 1\}$ denote the indicator function for $f(x) = i$.
For any $k, R \geq 2$, set parameters
$\rho = 1/\sqrt{(k - 1)\log R}$,
$d = 10k\log R$,
and $\delta = 1/(R^{10 + 100 k \log(R)})$.
Then there exists a $k$-query nonadaptive Dictator-vs.Quasirandom test with the following guarantees:
\begin{itemize}
    \item {\bf (Completeness)} If $f$ is a dictator, i.e. $f(x) = x_j$ for some
        coordinate $j \in [n]$, then the test passes with probability at least
        $$\rho^k = 1/((\log R)^{k/2}2^{O(k \log k)})$$

    \item {\bf (Soundness)} If $f$ has $Inf^{\leq d}_j[f^i] \leq \delta$ for all
        coordinates $j \in [n]$ and all projections $i \in [R]$,
        then the test passes with probability at most
       $$2^{O(k)} / R^{k-1}$$
\end{itemize}
\end{theorem}

Notice that if we assume $f$ is balanced, then this theorem is immediately implied
by the techniques of Section~\ref{sec:inapprox}.
However, to extend this to general functions via ``folding'',
we must technically show that the operation of folding keeps
low-influence functions as low-influence.
The full proof can be found in Appendix~\ref{sec:dtest-proof}.

\section{Conclusions and Open Questions} \label{sec:lim}

We conclude by posting interesting open questions regarding the approximability of {\sc Max $k$-CSP}$_R$ and providing our opinions on each question. First, as stated earlier, even with our results, current inapproximability results do not match the best known approximation ratio achievable in polynomial time when $3 \leq k < R$. Hence, it is intriguing to ask what the right ratio that {\sc Max $k$-CSP}$_R$ becomes NP-hard to approximate is. Since our hardness factor $2^{O(k \log k)}(\log R)^{k/2}/R^{k-1}$ does not match Chan's hardness factor $O(k/R^{k-2})$ when $k = R$, it is likely that there is a $k$ between 3 and $R - 1$ such that a drastic change in the hardness factor, and technique that yields that factor, occurs.

Moreover, since our PCP has completeness of $1/(2^{O(k)}(\log R)^{k/2})$, even if one cannot improve on the inapproximability factor, it is still interesting if one can come up with a hardness result with almost perfect completeness. In fact, even for $k = 2$, there is no known hardness of approximation of factor better than $O(\log R/\sqrt{R})$ with near perfect completeness whereas the best UGC-hardness known is $O(\log R/R)$.

It is also interesting to try to relax assumptions for other known inapproximability results from UGC to the One-Sided UGC. Since the One-Sided UGC is implied by $d$-to-1 Games Conjecture, doing so will imply inapproximability results based on the $d$-to-1 Games Conjecture. Moreover, without going into too much detail, we remark that most attempts to refute the UGC and the $d$-to-1 Conjecture need the value of the game to be high~\cite{ABS15, CMM06, ChMM06, GT06, Khot02, Kol11, T08}. Hence, these algorithms are not candidates to refute the One-Sided UGC. In addition, Arora, Barak and Steurer's~\cite{ABS15} subexponential time algorithm for unique games suggest that unique games have {\em intermediate complexity}, meaning that, even if the UGC is true, the UGC-hardness would not imply exponential time lower bounds. 
On the other hand, to the best of the authors' knowledge, the ABS algorithm does not run in subexponential time when the completeness is small. Hence, the One-Sided UGC may require exponential time to solve, which could give similar running time lower bounds for the resulting hardness of approximation results. Finally, there are evidences suggesting that relaxing completeness or soundness conditions of a conjecture can make it easier; the most relevant such result is that from Feige and Reichman who proved that, if one only cares about the approximation ratio and not completeness and soundness, then unique game is hard to approximate to within factor $\varepsilon$ for any $\varepsilon > 0$ ~\cite{FR04}.

\bibliographystyle{alpha}
\bibliography{k-csp}

\begin{thebibliography}{KKMO07}

\bibitem[ABS15]{ABS15}
Sanjeev Arora, Boaz Barak, and David Steurer.
\newblock Subexponential algorithms for unique games and related problems.
\newblock {\em J. {ACM}}, 62(5):42, 2015.

\bibitem[AM08]{AM08}
P.~Austrin and E.~Mossel.
\newblock Approximation resistant predicates from pairwise independence.
\newblock In {\em Computational Complexity, 2008. CCC '08. 23rd Annual IEEE
  Conference on}, pages 249--258, June 2008.

\bibitem[Cha13]{Chan13}
Siu~On Chan.
\newblock Approximation resistance from pairwise independent subgroups.
\newblock In {\em Proceedings of the Forty-fifth Annual ACM Symposium on Theory
  of Computing}, STOC '13, pages 447--456, New York, NY, USA, 2013. ACM.

\bibitem[CHK11]{CHK11}
Moses Charikar, MohammadTaghi Hajiaghayi, and Howard Karloff.
\newblock Improved approximation algorithms for label cover problems.
\newblock {\em Algorithmica}, 61(1):190--206, 2011.

\bibitem[CMM06a]{CMM06}
Moses Charikar, Konstantin Makarychev, and Yury Makarychev.
\newblock Near-optimal algorithms for unique games.
\newblock In {\em Proceedings of the Thirty-eighth Annual ACM Symposium on
  Theory of Computing}, STOC '06, pages 205--214, New York, NY, USA, 2006. ACM.

\bibitem[CMM06b]{ChMM06}
Eden Chlamtac, Konstantin Makarychev, and Yury Makarychev.
\newblock How to play unique games using embeddings.
\newblock In {\em 47th Annual {IEEE} Symposium on Foundations of Computer
  Science {(FOCS} 2006), 21-24 October 2006, Berkeley, California, USA,
  Proceedings}, pages 687--696, 2006.

\bibitem[CMM09]{CMM09}
Moses Charikar, Konstantin Makarychev, and Yury Makarychev.
\newblock Near-optimal algorithms for maximum constraint satisfaction problems.
\newblock {\em {ACM} Transactions on Algorithms}, 5(3), 2009.

\bibitem[CST01]{CST96}
Pierluigi Crescenzi, Riccardo Silvestri, and Luca Trevisan.
\newblock On weighted vs unweighted versions of combinatorial optimization
  problems.
\newblock {\em Information and Computation}, 167(1):10--26, 2001.

\bibitem[DMR09]{DMR09}
Irit Dinur, Elchanan Mossel, and Oded Regev.
\newblock Conditional hardness for approximate coloring.
\newblock {\em SIAM Journal on Computing}, 39(3):843--873, 2009.

\bibitem[DS10]{DS10}
Irit Dinur and Igor Shinkar.
\newblock On the conditional hardness of coloring a 4-colorable graph with
  super-constant number of colors.
\newblock In {\em Proceedings of the 13th International Conference on
  Approximation, and 14 the International Conference on Randomization, and
  Combinatorial Optimization: Algorithms and Techniques}, APPROX/RANDOM'10,
  pages 138--151, Berlin, Heidelberg, 2010. Springer-Verlag.

\bibitem[EH08]{EH08}
Lars Engebretsen and Jonas Holmerin.
\newblock More efficient queries in pcps for np and improved approximation
  hardness of maximum csp.
\newblock {\em Random Struct. Algorithms}, 33(4):497--514, December 2008.

\bibitem[Eng05]{Eng05}
Lars Engebretsen.
\newblock The nonapproximability of non-boolean predicates.
\newblock {\em SIAM J. Discret. Math.}, 18(1):114--129, January 2005.

\bibitem[FR04]{FR04}
Uriel Feige and Daniel Reichman.
\newblock On systems of linear equations with two variables per equation.
\newblock In Klaus Jansen, Sanjeev Khanna, JoséD.P. Rolim, and Dana Ron,
  editors, {\em Approximation, Randomization, and Combinatorial Optimization.
  Algorithms and Techniques}, volume 3122 of {\em Lecture Notes in Computer
  Science}, pages 117--127. Springer Berlin Heidelberg, 2004.

\bibitem[GM15]{GM15}
Gil Goldshlager and Dana Moshkovitz.
\newblock Approximating k{CSP} for large alphabets.
\newblock
  \url{https://people.csail.mit.edu/dmoshkov/papers/Approximating%20MAX%20kCSP.pdf},
  2015.

\bibitem[GR08]{GR08}
Venkatesan Guruswami and Prasad Raghavendra.
\newblock Constraint satisfaction over a non-boolean domain: Approximation
  algorithms and unique-games hardness.
\newblock In Ashish Goel, Klaus Jansen, JoséD.P. Rolim, and Ronitt Rubinfeld,
  editors, {\em Approximation, Randomization and Combinatorial Optimization.
  Algorithms and Techniques}, volume 5171 of {\em Lecture Notes in Computer
  Science}, pages 77--90. Springer Berlin Heidelberg, 2008.

\bibitem[GS11]{GS11}
Venkatesan Guruswami and Ali~Kemal Sinop.
\newblock The complexity of finding independent sets in bounded degree
  (hyper)graphs of low chromatic number.
\newblock In {\em Proceedings of the Twenty-Second Annual {ACM-SIAM} Symposium
  on Discrete Algorithms, {SODA} 2011, San Francisco, California, USA, January
  23-25, 2011}, pages 1615--1626, 2011.

\bibitem[GT06]{GT06}
Anupam Gupta and Kunal Talwar.
\newblock Approximating unique games.
\newblock In {\em Proceedings of the Seventeenth Annual {ACM-SIAM} Symposium on
  Discrete Algorithms, {SODA} 2006, Miami, Florida, USA, January 22-26, 2006},
  pages 99--106, 2006.

\bibitem[H{\aa}s96]{hastad1996}
Johan H{\aa}stad.
\newblock Clique is hard to approximate within $n^{1-\epsilon}$.
\newblock In {\em Foundations of Computer Science, 1996. Proceedings., 37th
  Annual Symposium on}, pages 627--636. IEEE, 1996.

\bibitem[Has05]{Hast05}
Gustav Hast.
\newblock Approximating {Max kCSP} - outperforming a random assignment with
  almost a linear factor.
\newblock In {\em Proceedings of the 32Nd International Conference on Automata,
  Languages and Programming}, ICALP'05, pages 956--968, Berlin, Heidelberg,
  2005. Springer-Verlag.

\bibitem[IP01]{IP01}
Russell Impagliazzo and Ramamohan Paturi.
\newblock On the complexity of k-sat.
\newblock {\em J. Comput. Syst. Sci.}, 62(2):367--375, 2001.

\bibitem[Kho02]{Khot02}
Subhash Khot.
\newblock On the power of unique 2-prover 1-round games.
\newblock In {\em Proceedings of the Thiry-fourth Annual ACM Symposium on
  Theory of Computing}, STOC '02, pages 767--775, New York, NY, USA, 2002. ACM.

\bibitem[Kho10]{Khot10}
Subhash Khot.
\newblock On the unique games conjecture (invited survey).
\newblock In {\em Proceedings of the 25th Annual {IEEE} Conference on
  Computational Complexity, {CCC} 2010, Cambridge, Massachusetts, June 9-12,
  2010}, pages 99--121, 2010.

\bibitem[KKMO07]{KKMO07}
Subhash Khot, Guy Kindler, Elchanan Mossel, and Ryan O'Donnell.
\newblock Optimal inapproximability results for {MAX-CUT} and other 2-variable
  csps?
\newblock {\em {SIAM} J. Comput.}, 37(1):319--357, 2007.

\bibitem[KKT15]{KKT15}
Guy Kindler, Alexandra Kolla, and Luca Trevisan.
\newblock Approximation of non-boolean 2csp.
\newblock {\em CoRR}, abs/1504.00681, 2015.

\bibitem[Kol11]{Kol11}
Alexandra Kolla.
\newblock Spectral algorithms for unique games.
\newblock {\em Computational Complexity}, 20(2):177--206, 2011.

\bibitem[KR08]{KR08}
Subhash Khot and Oded Regev.
\newblock Vertex cover might be hard to approximate to within 2-{$\epsilon$}.
\newblock {\em J. Comput. Syst. Sci.}, 74(3):335--349, May 2008.

\bibitem[MM14]{MM14}
Konstantin Makarychev and Yury Makarychev.
\newblock Approximation algorithm for non-boolean max-\emph{k}-csp.
\newblock {\em Theory of Computing}, 10:341--358, 2014.

\bibitem[MOO10]{MOO10}
Elchanan {Mossel}, Ryan {O'Donnell}, and Krzysztof {Oleszkiewicz}.
\newblock {Noise stability of functions with low influences: invariance and
  optimality.}
\newblock {\em {Ann. Math. (2)}}, 171(1):295--341, 2010.

\bibitem[O'D14]{OD14}
Ryan O'Donnell.
\newblock {\em Analysis of Boolean Functions}.
\newblock Cambridge University Press, 2014.

\bibitem[OW09a]{3bit}
Ryan O'Donnell and Yi~Wu.
\newblock 3-bit dictator testing: 1 vs. 5/8.
\newblock In {\em Proceedings of the Twentieth Annual {ACM-SIAM} Symposium on
  Discrete Algorithms, {SODA} 2009, New York, NY, USA, January 4-6, 2009},
  pages 365--373, 2009.

\bibitem[OW09b]{OW09}
Ryan O'Donnell and Yi~Wu.
\newblock Conditional hardness for satisfiable 3-csps.
\newblock In {\em Proceedings of the 41st Annual {ACM} Symposium on Theory of
  Computing, {STOC} 2009, Bethesda, MD, USA, May 31 - June 2, 2009}, pages
  493--502, 2009.

\bibitem[Rag08]{Rag08}
Prasad Raghavendra.
\newblock Optimal algorithms and inapproximability results for every csp?
\newblock In {\em Proceedings of the Fortieth Annual ACM Symposium on Theory of
  Computing}, STOC '08, pages 245--254, New York, NY, USA, 2008. ACM.

\bibitem[ST00]{ST00}
Alex Samorodnitsky and Luca Trevisan.
\newblock A {PCP} characterization of {NP} with optimal amortized query
  complexity.
\newblock In {\em Proceedings of the Thirty-Second Annual {ACM} Symposium on
  Theory of Computing, May 21-23, 2000, Portland, OR, {USA}}, pages 191--199,
  2000.

\bibitem[ST06]{ST06}
Alex Samorodnitsky and Luca Trevisan.
\newblock Gowers uniformity, influence of variables, and pcps.
\newblock In {\em Proceedings of the Thirty-eighth Annual ACM Symposium on
  Theory of Computing}, STOC '06, pages 11--20, New York, NY, USA, 2006. ACM.

\bibitem[Tan09]{Tan09}
Linqing Tang.
\newblock Conditional hardness of approximating satisfiable max 3csp-q.
\newblock In Yingfei Dong, Ding-Zhu Du, and Oscar Ibarra, editors, {\em
  Algorithms and Computation}, volume 5878 of {\em Lecture Notes in Computer
  Science}, pages 923--932. Springer Berlin Heidelberg, 2009.

\bibitem[Tre98]{Tre98}
Luca Trevisan.
\newblock Parallel approximation algorithms by positive linear programming.
\newblock {\em Algorithmica}, 21(1):72--88, 1998.

\bibitem[Tre08]{T08}
Luca Trevisan.
\newblock Approximation algorithms for unique games.
\newblock {\em Theory of Computing}, 4(1):111--128, 2008.

\end{thebibliography}

\appendix

\section{Proofs of Preliminary Results}

For completeness, we prove some of the preliminary results,
whose formal proofs were not found in the literature by the authors.

\subsection{Mollification Lemma} \label{sec:apx-mollification}

Below is the proof of the Mollification Lemma. We remark that, while
its main idea is explained in~\cite{OD14}, the full proof is not shown there.
Hence, we provide the proof here for completeness.

\begin{proof}(of Lemma~\ref{lem:mollify})
Let $p: \R \to \R$ be a $\mathcal{C}^4$ function supported only on $[-1, +1]$,
such that $p(y)$ forms a probability distribution.
(For example, an appropriately normalized version of
$e^{-1/(1+y^2)}$ for $|y| \leq 1$).
Define $p_\lambda(y)$ to be re-scaled to have support $[-\lambda, +\lambda]$
for some $\lambda > 0$:
\beq
p_\lambda(y) := (1/\lambda) p(y / \lambda).
\eeq
Let $Y_\lambda$ be a random variable with distribution $p_\lambda(y)$,
supported on $[-\lambda, +\lambda]$.
We will set $\lambda = \zeta / c$.

Now, define
\beq
\t \psi := \E_{Y_\lambda}[\psi(x + Y_\lambda)].
\eeq

This is pointwise close to $\psi$, since $\psi$ is $c$-Lipschitz:
\beq
    |\t \psi(x)  - \psi(x)| = |\E_{Y_\lambda}[\psi(x + Y_\lambda)- \psi(x)]|
\leq \E_{Y_\lambda}[|\psi(x + Y_\lambda) - \psi(x)|]
   \leq \E_{Y_\lambda}[c|Y_\lambda|]
    \leq c \lambda = \zeta.
\eeq

Further, $\t \psi$ is $\mathcal{C}^3$, because $\t \psi(x)$ can be written as a
convolution:
\beq
\t\psi(x) = (\psi \ast p_\lambda)(x)
\implies
\t\psi''' = (\psi \ast p_\lambda)''' = (\psi \ast p_\lambda''').
\eeq
To see that $\t \psi'''$ is bounded,
for a fixed $x \in \R$, define the constant $z := \psi(x)$. Then,
\begin{align}
|\t\psi'''(x)|
&= |(\psi \ast p_\lambda''')(x)|\\
(z \text{ is constant, so } z'=0)\quad
&= |(\psi \ast p_\lambda''' - z' \ast p_\lambda'')(x)|\\
&= |(\psi \ast p_\lambda''' - z \ast p_\lambda''')(x)|\\
&= |((\psi - z) \ast p_\lambda''')(x)|\\
&= \left|\int_{-\infty}^{+\infty} p_\lambda'''(y) (\psi(x - y)  - z)dy \right|\\
&= \left|\int_{-\infty}^{+\infty} p_\lambda'''(y) (\psi(x - y)  - \psi(x))dy \right|\\
&\leq \int_{-\lambda}^{+\lambda} |p_\lambda'''(y)| |\psi(x - y) - \psi(x)| dy\\
(c\text{-Lipschitz})\quad
&\leq  ||p_\lambda'''||_\infty \int_{-\lambda}^{+\lambda}|cy| dy\\
&=  ||p_\lambda'''||_\infty c\lambda^2.
\end{align}

Define the universal constant $\t C := ||p'''||_\infty$.
We have
\beq
p_\lambda'''(y) = (1/\lambda^4) p'''(y/\lambda)
\implies
||p_\lambda'''||_\infty \leq (1/\lambda^4) \t C.
\eeq

With our choice of $\lambda = \zeta / c$, this yields $|\t\psi'''(x)| \leq \t C c^3 / \zeta^2$, which completes the proof of Lemma~\ref{lem:mollify}.
\qed
\end{proof}

\subsection{Proof of Lemma~\ref{ourinv}} \label{sec:app-ourinv}

Below we show the proof of Lemma~\ref{ourinv}.

\proof
First, we ``mollify'' the function $\psi$
to construct a $\mathcal{C}^3$ function $\t \psi$,
by applying Lemma~\ref{lem:mollify} for $\zeta = 1/R^k$.
Notice that both choices of $\psi$ are $k$-Lipschitz.
Therefore the Mollification Lemma guarantees that
$|\t \psi'''(x)| \leq \t C k^3 R^{2k}$
for some universal constant $\t C$.

Since $\t \psi$ is pointwise close to $\psi$, with deviation at most $1/R^k$, we have
\begin{equation}
\left| \E_{y \in \{\pm 1\}^{nR}}[\psi(\Fk(y))]
- \E_{x \in [R]^{n}}[\psi(\fk(x))] \right|
\leq
\left| \E_{y \in \{\pm 1\}^{nR}}[\t \psi(\Fk(y))]
- \E_{x \in [R]^{n}}[\t \psi(\fk(x))] \right|
+
O(1/R^k).
\end{equation}

Applying the General Invariance Principle (Theorem~\ref{thm:maininv}) with the function
$\t \psi$, we have
\begin{equation}
\label{eqn:use_inv}
\left| \E_{y \in \{\pm 1\}^{nR}}[\t \psi(\Fk(y))]
- \E_{x \in [R]^{n}}[\t \psi(\fk(x))] \right|
\leq \t C k^3R^{2k} 10^d R^{d/2} \sqrt{\delta}.
\end{equation}
By our choice of parameters $d, \delta$, this is $O(1/R^k)$.
\qed

\section{$d$-to-1 Games Conjecture implies One-Sided Unique Games Conjecture} \label{sec:dto1osugc}

In this section, we prove that if $d$-to-1 Games Conjecture is true, then so is One-Sided Unique Games Conjecture.

\begin{lemma}
For every $d \in \mathbb{N}$, $d$-to-1 Games Conjecture implies One-Sided UGC.
\end{lemma}

\begin{proof}
Suppose that $d$-to-1 Games Conjecture is true for some $d \in \mathbb{N}$. We will prove One-Sided UGC; more specifically, $\zeta$ in the One-Sided UGC is $1/d$.
The reduction from a $d$-to-1 game $(V, W, E, N, \{\pi_e\}_{e \in E})$ to a unique game $(V', W', E', N', \{\pi'_e\}_{e \in E})$ can be described as follows:
\begin{itemize}
\item Let $V' = V, W' = W, E' = E$, and $N' = N$
\item We define $\pi'_e$ as follows. For each $\theta \in [N/d]$, let the elements of $\pi^{-1}_e(\theta)$ be $\sigma_1, \sigma_2, \dots, \sigma_d \in [N]$. We then define $\pi'_e(\sigma_i) = d(\theta - 1) + i$.
\end{itemize}

Now, we will prove the soundness and completeness of this reduction.

{\bf (Completeness)} Suppose that the $d$-to-1 game is satisfiable. Let $\varphi: V \cup W \rightarrow [N]$ be the assignment that satisfies every constraint in the $d$-to-1 game. We define $\varphi': V' \cup W' \rightarrow [N']$ by first assign $\varphi'(v) = \varphi(v)$ for every $v \in V$. Then, for each $w \in W$, pick $\varphi'(w)$ to be an assignment that satisfies as many edges touching $w$ in the unique game as possible, i.e., for a fixed $w$, $\varphi'(w)$ is select to maximize $|\{v \in N(w) \mid \pi'_{(v, w)}(\varphi(v)) = \varphi'(w)\}|$ where $N(w)$ is the set of neighbors of $w$. From how $\varphi'(w)$ is picked, we have
\begin{align}
|\{v \in N(w) \mid \pi'_{(v, w)}(\varphi(v)) = \varphi'(w)\}| &\geq \frac{1}{d} \sum_{i=1}^{d} |\{v \in N(w) \mid \pi'_{(v, w)}(\varphi(v)) = d(\varphi(w) - 1) + i\}|.
\end{align}

Let $1[{\pi'_{(v, w)}(\varphi(v)) = d(\varphi(w) - 1) + i}]$ be the indicating variable whether ${\pi'_{(v, w)}(\varphi(v)) = d(\varphi(w) - 1) + i}$, we can rewrite the right hand side as follows:

\begin{align}
&\frac{1}{d} \sum_{i=1}^{d} |\{v \in N(w) \mid \pi'_{(v, w)}(\varphi(v)) = d(\varphi(w) - 1) + i\}| \\
&= \frac{1}{d} \sum_{i=1}^{d} \sum_{v \in N(w)} 1[{\pi'_{(v, w)}(\varphi(v)) = d(\varphi(w) - 1) + i}] \\
&= \frac{1}{d} \sum_{v \in N(w)} \sum_{i=1}^{d} 1[{\pi'_{(v, w)}(\varphi(v)) = d(\varphi(w) - 1) + i}].
\end{align}

From how $\pi'_{(v, w)}$ is defined and since $\pi_{(v, w)}(\varphi(v)) = \varphi(w)$, there exists $i \in [d]$ such that $\pi'_{(v, w)}(\varphi(v)) = d(\varphi(w) - 1) + i$. As a result, we have
\begin{align}
\frac{1}{d} \sum_{v \in N(w)} \sum_{i=1}^{d} 1[{\pi'_{(v, w)}(\varphi(v)) = d(\varphi(w) - 1) + i}] \geq \frac{1}{d} \sum_{v \in N(w)} 1 = \frac{|N(w)|}{d}.
\end{align}

In other words, at least $1/d$ fraction of edges touching $w$ is satisfied in the unique game for every $w \in W$. Hence, $\varphi'$ has value at least $1/d$, which means that the unique game also has value at least $1/d$.

{\bf (Soundness)} Suppose that the value of the $d$-to-1 game is at most $\gamma$. For any assignment $\varphi': V' \cup W' \rightarrow [N']$ to the unique game, we can define an assignment $\varphi: V \cup W \rightarrow [N]$ by
\begin{align}
\varphi(u) =
\begin{cases}
\varphi'(u) & \text{if } u \in V, \\
\lfloor (\varphi'(u) - 1)/d \rfloor + 1 &\text{ if } u \in W.
\end{cases}
\end{align}

From how $\pi'_e$ is defined, it is easy to see that, if $\pi'_e(\varphi'(v)) = \varphi'(w)$, then $\pi_e(\varphi(v)) = \varphi(w)$. In other words, the value of $\varphi'$ with respect to the unique game is no more than the value of $\varphi$ with respect to the $d$-to-1 game. As a result, the value of the unique game is at most $\epsilon$.

As a result, if it is NP-hard to distinguish a satisfiable $d$-to-1 game from one with value at most $\gamma$, then it is also NP-hard to distinguish a unique game of value at least $\zeta = 1/d$ from that with value at most $\gamma$, which concludes the proof of this lemma.
\qed
\end{proof}

\section{Proof of Dictator Test}
\label{sec:dtest-proof}

Here we prove our result for the Dictator-vs.-Quasirandom test (Theorem~\ref{thm:dtest}).

\proof (of Theorem~\ref{thm:dtest})
For $c \in [R]$, define the function
\beq
f_c(x_1, x_2, \dots, x_n) := f(x_1 + c, x_2 + c, \dots, x_n + c) - c.
\eeq
Note that $\pm c$ is performed modulo $R$.

The test works as follows:
Pick $z \in [R]^n$ uniformly at random,
and let $\i{x}{1}, \i{x}{2}, \dots, \i{x}{k}$ be independent $\rho$-correlated noisy copies of
$z$.
Then, pick $c_1, c_2, \dots, c_k$ independently uniformly at random, where each $c_i \in [R]$.
Accept iff
\beq
f_{c_1}(\i{x}{1}) = f_{c_2}(\i{x}{2}) = \dots = f_{c_k}(\i{x}{k}).
\eeq

For completeness, notice that if $f$ is a dictator, then $f_c = f$ for all $c \in [R]$.
Say $f$ is a dictator on the $j$-th coordinate: $f(x) = x_j$.
Then the test clearly accepts with probability at least $\rho^k$ (if none of the
coordinates $j$ were perturbed in all the noisy copies $\i{x}{i} \pcorr z$).

For soundness: For any $i \in [R]$, let $f^i: [R]^n \to \{0, 1\}$
denote the indicator function for $f(x) = i$, and similarly for $f_c^i: [R]^n
\to \{0, 1\}$.
Notice that
\beq
f^i_c(x) = f^{i+c}(x + (c, c, \dots, c))
\eeq

Then, write the acceptance probability as
\bal
\Pr[accept] &=
\Pr_{c_i, z, \i{x}{j} \pcorr z}[f_{c_1}(\i{x}{1}) = f_{c_2}(\i{x}{2}) = \dots = f_{c_k}(\i{x}{k})]\\
&=
\sum_{i \in [R]}
\Pr_{c_i, z, \i{x}{j} \pcorr z}[i = f_{c_1}(\i{x}{1}) = f_{c_2}(\i{x}{2}) = \dots = f_{c_k}(\i{x}{k})]\\
&=
\sum_{i \in [R]}
\E_{c_i, z, \i{x}{j} \pcorr z}[f_{c_1}^i(\i{x}{1}) f_{c_2}^i(\i{x}{2}) \dots f_{c_k}^i(\i{x}{k})]\\
(\text{Independence of } c_i) \quad &=
\sum_{i \in [R]}
\E_{z, \i{x}{j}\pcorr z}[\E_{c_1}[f_{c_1}^i(\i{x}{1})] \E_{c_2}[f_{c_2}^i(\i{x}{2})] \dots
\E_{c_2}[f_{c_k}^i(\i{x}{k})]].\\
\eal

If we define the function $g^i: [R]^n \to [0, 1]$ as
\beq
g^i(x) := \E_c[f^i_c(x)].
\eeq
Then this acceptance probability is
\bal
\Pr[accept] &=
\sum_{i \in [R]}
\E_{z, x_i \pcorr z}[g^i(\i{x}{1})g^i(\i{x}{2})\dots g^i(\i{x}{k})]\\
&=
\sum_{i \in [R]}
\E_{z}[(T_\rho g^i(z))^k].
\label{eqn:dtest-hyper}
\eal
Notice that $\E_x[g^i(x)] = 1/R$, because
\bal
\E_x[g^i(x)] = \E_{x, c}[f^i_c(x)]
&= \E_{x, c}[f^{i+c}(x+ (c, c, \dots, c))]\\
\text{($c, x$ same joint distribution as $i+c, x+c$)} \quad
&= \E_{x, c}[f^{c}(x)]\\
&= \E_x[\E_c[f^c(x)]]
= \E_x[1/R] = 1/R.
\eal
Thus, if the function $g^i$ has small low-degree influences, then
Lemma~\ref{mainlemma} (Main Lemma) applied to $g^i$
in line~\eqref{eqn:dtest-hyper}
directly implies that this acceptance
probability is $2^{O(k)} / R^{k-1}$.
We will now formally show that the influences of the
``expected folded function'' $g^i$ are bounded by the influences of the original $f^i$.

First, the Fourier coefficients of $g^i$ are
\bal
\hat g^i(s) &= \E_c[\hat f^i_c(s)].
\eal
Thus the low-degree influences of $g^i$ are bounded as
\bal
Inf_j^{\leq d}[g^i]
&=
\sum_{s \in [R]^n \atop s(j) \neq 1, |s| \leq d}
\hat g^i(s)^2 \\
&=
\sum_{s \in [R]^n \atop s(j) \neq 1, |s| \leq d}
\E_c[\hat f^i_c(s)]^2\\
&\leq
\sum_{s \in [R]^n \atop s(j) \neq 1, |s| \leq d}
\E_c[\hat f^i_c(s)^2]\\
&=
\E_c[\sum_{s \in [R]^n \atop s(j) \neq 1, |s| \leq d}
\hat f^i_c(s)^2]\\
&=
\E_c[Inf_j^{\leq d}[f_c^i]].
\label{eqn:avginfc}
\eal

Finally, we must relate the influences of $f_c^i$ to the influences of $f^i$.
For a fixed $c \in [R]$, we have
\bal
Inf_j^{\leq d}[f_c^i]
&=
Inf_j[(f_c^i)^{\leq d}]\\
&=
\E_{x \in [R]^n}[Var_{x_j \in [R]}[(f_c^i)^{\leq d}]] \\
&=
\E_{x \in [R]^n}[Var_{x_j \in [R]}[(f^{i+c})^{\leq d}(x_1 + c, x_2 + c, \dots, x_n +c)]] \\
&=
\E_{x \in [R]^n}[Var_{x_j \in [R]}[(f^{i+c})^{\leq d}(x_1, x_2, \dots, x_n)]] \\
&=
Inf_j[(f^{i+c})^{\leq d}]\\
&=
Inf_j^{\leq d}[f^{i+c}].
\label{eqn:infic}
\eal
Therefore, if $Inf^{\leq d}_j[f^i] \leq \delta$ for all coordinates $j \in [n]$
and all projections $i \in [R]$ (as we assume for soundness), then
from~\eqref{eqn:avginfc} and~\eqref{eqn:infic} we have
\beq
Inf_j^{\leq d}[g^i] \leq
\E_c[Inf_j^{\leq d}[f_c^i]]
=
\E_c[ Inf_j^{\leq d}[f^{i+c}]]
\leq \delta.
\eeq
Thus the function $g^i$ has small low-degree influences as well.

So we can complete the proof, continuing from line~\eqref{eqn:dtest-hyper} and
applying our Main Lemma to $g^i$:
\bal
\Pr[accept] &=
\sum_{i \in [R]}
\E_{z}[(T_\rho g^i(z))^k]\\
\text{(Lemma~\ref{mainlemma})}\quad &\leq
\sum_{i \in [R]}
2^{O(k)}/R^k\\
&=
2^{O(k)}/R^{k-1}.
\eal
\qed

\end{document}